\documentclass[a4paper,UKenglish,cleveref,autoref,thm-restate,nolineno]{socg-lipics-v2021}

\usepackage{subfiles}
\usepackage[utf8]{inputenc}
\usepackage[T1]{fontenc}
\usepackage[numbers]{natbib}
\usepackage[algo2e, ruled, vlined]{algorithm2e}
\usepackage{algorithm}
\usepackage{lmodern}
\usepackage{bm}
\usepackage{array}
\usepackage{adjustbox}
\usepackage{amsfonts} 
\usepackage{amssymb}  
\usepackage{amsthm}
\usepackage{amsmath} 
\usepackage{bookmark}
\usepackage{booktabs}
\usepackage{color}    
\usepackage{graphicx} 
\usepackage{hyperref} 
\usepackage{url}
\usepackage{mathtools}
\usepackage{listings}
\usepackage{verbatim}
\usepackage{todonotes}
\usepackage{pgfplots}
\pgfplotsset{compat=1.18}
\usepackage{caption}
\usepackage{subcaption}
\usepackage{float}
\usepackage[table]{xcolor}
\usepackage{xstring}
\usepackage{numprint}
\usepackage{pdflscape}
\usepackage{datatool}
\usepackage{longtable}
\usepackage{multicol}
\usepackage{wasysym} 

\npthousandsep{ }
\npdecimalsign{.}
\makeatletter
\let\@ET@sw@cr\cr
\makeatother

\makeatletter
\newcommand{\subfilesbibliography}[1]{%
   \expandafter\ifx\csname ver@subfiles.cls\endcsname\relax
   \expandafter\@secondoftwo
   \else
       \expandafter\@firstoftwo
   \fi
   {\bibliography{#1}}
 {}%
}
\makeatother

\newcommand{\etal}{et al.}

\newcommand{\strong}{\textsc{Red}}
\newcommand{\ilpg}{\textsc{ILP$_G$}}
\newcommand{\pilpg}{\textsc{\strong-ILP$_G$}}
\newcommand{\ilp}{\textsc{ILP}}
\newcommand{\pilp}{\textsc{\strong-ILP}}
\newcommand{\satAndReduce}{\textsc{SatAndReduce}}
\newcommand{\psatAndReduce}{\textsc{\strong-satAndReduce}}
\newcommand{\weGotYouCovered}{\textsc{WeGotYouCovered}}
\newcommand{\pweGotYouCovered}{\textsc{\strong-weGotYouCovered}}
\newcommand{\struction}{\textsc{Struction}}
\newcommand{\pstruction}{\textsc{\strong-struction}}

\usepackage{tcolorbox}
\tcbuselibrary{theorems}
\newtcbtheorem[number within=section]{obs}{Observation}{
	colback=lightgray, colframe=lightgray, boxsep=0pt, left=1mm, right=1mm, top=5pt,  bottom=5pt, 	tcbox raise base, fonttitle=\bfseries, coltitle=black, coltext=black,  terminator sign={}, attach title to upper, after title={\hspace{1em}}}{obs}

  \newcommand{\TimeStyle}[1]{%
  \edef\TimeTmp{#1}%
  \ifx\TimeTmp\empty
    --%
  \else
    \pgfmathparse{#1<0 ? 1 : 0}%
    \ifnum\pgfmathresult=1
      -%
    \else
      \pgfmathparse{#1<0.01 ? 1 : 0}%
      \ifnum\pgfmathresult=1
        $<0.01$%
      \else
        \pgfmathprintnumber[
          fixed,
          precision=2,
          zerofill,
          1000 sep={\,},          
          assume math mode=false 
        ]{#1}%
      \fi
    \fi
  \fi
}
\newcommand{\CompareTimesStyle}[2]{%
  \if\relax\detokenize{#1}\relax
    -- %
  \else
    \if\relax\detokenize{#2}\relax
      \textbf{\TimeStyle{#1}} 
    \else
      \pgfmathparse{#1<#2 ? 1 : 0}%
      \ifnum\pgfmathresult=1
        \boldmath{\TimeStyle{#1}} 
      \else
        \TimeStyle{#1} 
      \fi
    \fi
  \fi
}

\usepackage{pgf}
\usepackage{pgfplots}
\usepackage{pgfplotstable}
\usepackage{tikz} 

\usepgfplotslibrary{colorbrewer}
\pgfdeclarelayer{foreground}
\pgfsetlayers{main,foreground} 
\usepgfplotslibrary{statistics}
\usepgfplotslibrary{groupplots}
\usetikzlibrary{arrows.meta, positioning, tikzmark, calc, bending, shadows, shapes,patterns}
\pgfplotsset{compat=1.18}

\pgfplotstableset{
  timeStyle/.style={
	fixed, 
	precision=2, 
	1000 sep={\,},
	zerofill,
    string type=false,
	postproc cell content/.append code={
		\pgfmathparse{##1<0 ? 1 : 0}%
		\ifnum\pgfmathresult=1
			\pgfkeysalso{@cell content={-}}%
		\else
			\pgfmathparse{##1<0.01 ? 1 : 0}%
			\ifnum\pgfmathresult=1
				\pgfkeysalso{@cell content={$<0.01$}}%
			\else
				\pgfkeysalso{@cell content={\pgfmathprintnumber[fixed,precision=2]{##1}}}%
			\fi
		\fi
	}
  }
}

\pgfplotsset{
    cycle list/.define={my marks}{
        every mark/.append style={solid,fill=\pgfkeysvalueof{/pgfplots/mark list fill}},mark=*\\
        every mark/.append style={solid,fill=\pgfkeysvalueof{/pgfplots/mark list fill}},mark=square*\\
        every mark/.append style={solid,fill=\pgfkeysvalueof{/pgfplots/mark list fill}},mark=triangle*\\
        every mark/.append style={solid,fill=\pgfkeysvalueof{/pgfplots/mark list fill}},mark=halfcircle*\\
        every mark/.append style={solid,fill=\pgfkeysvalueof{/pgfplots/mark list fill}, mark options={rotate=90}, scale=1.25}, mark=halfsquare right*\\
        every mark/.append style={solid,fill=\pgfkeysvalueof{/pgfplots/mark list fill},rotate=180},mark=triangle*\\
        every mark/.append style={solid,fill=\pgfkeysvalueof{/pgfplots/mark list fill}, scale=1.25},mark=10-pointed star\\
        every mark/.append style={solid,fill=\pgfkeysvalueof{/pgfplots/mark list fill}, mark options={rotate=90}, scale=1.25}, mark=halfsquare left*\\
        every mark/.append style={solid,fill=\pgfkeysvalueof{/pgfplots/mark list fill},rotate=90},mark=triangle*\\
    },
	cycle list/Dark2,
    mark list fill={.!75!white},
    cycle multiindex* list={Dark2 \nextlist my marks \nextlist [3 of]linestyles \nextlist thick \nextlist},
	every axis/.style={
        axis lines=middle, 
        axis x line*=bottom,
        axis y line*=left,
		y axis line style={-},
		x axis line style={>=Stealth},
		width=0.44\textwidth,
		height=0.4\textwidth,
		xlabel={$\tau$},
		ylabel={Fraction of instances},
    	ytick={0,0.2,0.4,0.6,0.8,1}, 
    	ymin=0,
		scale only axis,
		xlabel style={at={(1,0)},anchor=west},
		ylabel near ticks
	},
	myLegend/.style={
		legend style={
			anchor=center,
			cells={anchor=west},
			draw=none,    
			/tikz/every even column/.append style={column sep=5mm}
		}
	},
	myLegend2/.style={
		myLegend,
		legend style={legend columns=2}
	},
	myLegend3/.style={
		myLegend,
		legend style={legend columns=3}
	},
	myLegend4/.style={
		myLegend,
		legend style={legend columns=4}
	},
	myLegend5/.style={
		myLegend,
		legend style={legend columns=5}
	},
	myLegend6/.style={
		myLegend,
		legend style={legend columns=6}
	},
	myLegend7/.style={
		myLegend,
		legend style={legend columns=7}
	},
	perf_max/.style={
        ymin=0, ymax=1, 
		x dir= reverse,
		xmax=1,
        x axis line style={Stealth-},
        y axis line style={-}
    },
    perf_min/.style={
        ymin=0, ymax=1,
		xmin=1,
        x axis line style={-Stealth},
        y axis line style={-}
    },
    perf_left/.style={
		width=0.22\textwidth,
        ylabel near ticks,
        x axis line style={-},
		xlabel={}
    },
    perf_mid/.style={
		width=0.22\textwidth,
        ytick={},
        yticklabels={},
		ylabel={},
        x axis line style={-},
		xlabel={}
    },
    perf_right/.style={
		width=0.22\textwidth,
        ytick={},
        yticklabels={},
		ylabel={},
        xlabel={$\tau$}
	},
	red2pack/.style={
        ymin=0, ymax=100,
		height=0.2\textwidth,
		ylabel={$foi$},
        ytick={0,20,40,60,80,100},
        yticklabels={0,0.2,0.4,0.6,0.8,1},
		scale only axis,
		ylabel style={at={(0,1.1)},anchor=center,rotate=-90},
		xlabel style={at={(1,0)},anchor=west}
	},
	oneofthree/.style={
		red2pack,
		width=0.275\textwidth
	},
	twoofthree/.style={
		red2pack,
		xshift=0.3625\textwidth,
		width=0.275\textwidth
	},
	threeofthree/.style={
		red2pack,
		xshift=0.725\textwidth,
		width=0.275\textwidth
	},
	single_plot/.style={
      width=0.6\textwidth,
      height=0.3\textwidth,
      legend style={at={(0.5,1.25)},anchor=north},
	}
}

\tikzset{
	graph/.style={draw, fill=remainingGColor, rounded corners=4.5mm, inner sep=2.5mm, align=center, minimum width=2.5cm},	
	hypergraph/.style={ black, fill=remainingGColor},
	node/.style={circle, draw, fill=black, inner sep=0pt, minimum size=6pt},
	lnode/.style={node, label=left:#1},
	rnode/.style={node, label=right:#1},
	tnode/.style={node, label=above:#1},
	bnode/.style={node, label=below:#1},
	nodeR/.style={ node, draw=lightgray, fill=lightgray},
	nodeE/.style={ node, draw=redgray, fill=redgray, fill opacity=0.8},
	nodeI/.style={ node, draw=greengray, fill=greengray, fill opacity=0.8},
	weight/.style={node, rectangle, fill opacity=0, draw opacity=0, text opacity = 1,text=weightColor, yshift=0.4cm, font=\small},
	lweight/.style={node, rectangle, fill opacity=0, draw opacity=0, text opacity = 1,text=weightColor, xshift=-0.5cm, font=\small},
	rweight/.style={node, rectangle, fill opacity=0, draw opacity=0, text opacity = 1,text=weightColor, xshift=0.5cm, font=\small},
	weightR/.style={node, rectangle, fill opacity=0, draw opacity=0, text opacity = 1,text=redgray, yshift=0.5cm, font=\small},
	weightText/.style={node, rectangle, fill opacity=1, text opacity = 1,text=weightColor, fill=white, draw=white , font=\small},
	edge/.style={draw=black, fill=black, thick},
	edgeR/.style={ edge, draw=lightgray, fill=lightgray, dashed},
	matching/.style={thick, draw=black, fill=green, fill opacity=0.4},
	Rmatching/.style={thick, draw=none, fill=lightgray, fill opacity=0.6},	
	Imatching/.style={thick, draw=none, fill=green, fill opacity=0.4},	
	Ematching/.style={thick, draw=none, fill=redgray, fill opacity=0.5}
}

\newtheorem{reduction}{Reduction}

\hideLIPIcs  

\title{Data Reductions for the Strong Maximum Independent Set Problem in Hypergraphs}

\author{Ernestine Gro{\ss}mann\footnote{corresponding author}}{Heidelberg University,  Faculty of Mathematics and Computer Science, Germany}{e.grossmann@informatik.uni-heidelberg.de}{https://orcid.org/0000-0002-9678-0253}{}
\author{Christian Schulz}{Heidelberg University, Faculty of Mathematics and Computer Science, Germany}{christian.schulz@informatik.uni-heidelberg.de}{https://orcid.org/0000-0002-2823-3506}{}
\author{Darren Strash}{Hamilton College, Department of Computer Science, USA}{dstrash@hamilton.edu}{0000-0001-7095-8749}{}
\author{Antonie Wagner}{Heidelberg University, Faculty of Mathematics and Computer Science, Germany}{antonie.wagner@stud.uni-heidelberg.de}{https://orcid.org/}{}

\authorrunning{E. Gro{\ss}mann, C. Schulz, D. Strash, A. Wagner} 

\Copyright{Ernestine Gro{\ss}mann, Christian Schulz, Darren Strash, and Antonie Wagner} 



\category{} 

\supplement{Our code is available on \url{https://github.com/KarlsruheMIS/HyperMIS}~\cite{source-code}}

\funding{This work was supported by the DFG grant SCHU 2567/8-1.}

\EventEditors{John Q. Open and Joan R. Access}
\EventNoEds{2}
\EventLongTitle{42nd Conference on Very Important Topics (CVIT 2016)}
\EventShortTitle{CVIT 2016}
\EventAcronym{CVIT}
\EventYear{2016}
\EventDate{December 24--27, 2016}
\EventLocation{Little Whinging, United Kingdom}
\EventLogo{}
\SeriesVolume{42}
\ArticleNo{23}

\begin{document}

\maketitle

\begin{abstract}
    This work addresses the well-known Maximum Independent Set problem in the context of hypergraphs.
    While this problem has been extensively studied on graphs, we focus on its strong extension to hypergraphs, where edges may connect any number of vertices. A set of vertices in a hypergraph is strongly independent if there is at most one vertex per edge in the set. One application for this problem is to find perfect minimal hash functions.
    We propose nine new data reduction rules specifically designed for this problem.
    Our reduction routine can serve as a preprocessing step for any solver.
    We analyze the impact on the size of the reduced instances and the performance of several subsequent solvers when combined with this preprocessing.
    Our results demonstrate a significant reduction in instance size and improvements in running time for subsequent solvers.
    The preprocessing routine reduces instances, on average, to $22\,\%$ of their original size in 6.76 seconds. When combining our reduction preprocessing with the best-performing exact solver, we observe an average speedup of 3.84x over not using the reduction rules. In some cases, we can achieve speedups of up to 53x.
    Additionally, one more instance becomes solvable by a method when combined \hbox{with our preprocessing.}
\end{abstract}
\newpage

\section{Introduction}
\label{sec:intro}

The Maximum Independent Set (MIS) problem is a fundamental NP-hard problem in combinatorial optimization.
Given a graph \(G = (V,E)\), a maximum independent set \(\mathcal{I} \subseteq V\) is a set of pairwise non-adjacent vertices with maximum cardinality.
This problem has a wide range of applications across various domains.
In cartography, it plays a crucial role in dynamic map labeling \cite{Gemsa16}, which involves determining optimal placements for labels associated with geographic features.
In computational biology, MIS is utilized in DNA computing and nanotechnology applications \cite{ma23,Yu10}, where independent sets help model molecular self-assembly and sequence selection.
Additionally, it serves as a powerful analytical tool in network science, particularly in the study of social interactions and influence within multilayer social networks \cite{Khomami24}.

While the MIS problem is traditionally studied in the context of graphs, many real-world problems exhibit higher-order relationships that cannot be captured by pairwise interactions alone.
In such cases, hypergraphs provide a more expressive framework, allowing edges to connect multiple vertices simultaneously.
Extending the MIS problem to hypergraphs yields the task of finding the largest subset of vertices such that no two selected vertices belong to the same hyperedge.

We highlight the following application of this problem which is the construction of perfect (minimal) hash functions. The goal is to assign $n$ keys a unique value between $0$ and $n-1$. This process typically works in two steps. First, all keys are grouped into $k$ buckets using an initial hash function. Second, for each of these buckets, a secondary hash function must be chosen that maps the keys in that bucket to specific output values between $0$ and $n-1$. The challenge is to choose secondary hash functions for each bucket so that no two buckets map keys to the same output value. This problem can be modeled using a hypergraph. Vertices correspond to a (bucket, hash function) pair, and edges connect all vertices in the same bucket, as well as vertices that would cause a collision, i.e., they are output elements from two buckets mapped to the same output value. If we can find an MIS of size $k$ in this hypergraph, then there is a perfect hash function. Each solution vertex specifies which hash function to use for each bucket.

A key strategy for tackling complex combinatorial optimization problems is the use of data reduction rules.
These techniques aim to simplify problem instances by removing redundant structures and reducing the input size while preserving the equivalence of the solution space.
Reduction methods have proven to be successful for the MIS problem in graphs, particularly in the context of exact algorithms and preprocessing for heuristics \cite{Bourgeois12,ChangLZ17,LammSSSW17,Strash16,TarjanT77}.
By identifying and eliminating fast solvable components, these approaches can significantly enhance the performance of solving methods, making them an essential tool for handling large-scale instances.
This success motivates the use of data reduction techniques as a preprocessing phase for solving the MIS problem in hypergraphs.

\textbf{Our Contributions.}
In this work, we introduce nine novel exact data reduction rules for the strong Maximum Independent Set problem in hypergraphs. With those reductions, we design a preprocessing routine to enhance the solvability of the problem by reducing instance size before applying subsequent solvers.
To evaluate the effectiveness of our approach, we present an extensive experimental evaluation. We analyze the reduction in instance size and the impact on the performance of subsequent solvers.
On average, our reduction routine can reduce our dataset to 22\,\% of its original size. For subsequent solvers this results in average speedups of 1.9x, while on single instances, we observe improvements in running time of up to a factor of 53x.
\subfilesbibliography{../bib}

\section{Preliminaries}
\label{sec:preliminaries}

An \textit{undirected hypergraph} $H = (V,E)$ is defined as a set of \(n\) vertices \(V\) and a set of \(m\) hyperedges \(E\), also referred to as edges.
Each edge \(e \in E\) is a subset of the vertex set \(V\), i.e. \(E \subseteq 2^V\).
We consider hyperedges as sets rather than multisets, allowing multiple edges to include the same set of vertices, while a vertex can appear in a hyperedge only \textit{once}.
The set of vertices of a hyperedge $e$ is described with \(e\) and \(|e|\) is the edge size.
A vertex is \textit{incident} to an edge \(e\) iff \(v \in e\). The set of edges incident to a vertex \(v\) is denoted by \mbox{\(E(v) = \{e \in E \mid v \in e\}\).}
The degree of a vertex \(v\) is \mbox{\(d(v) = |E(v)|\)}.
We additionally define the size $|H|$ of a hypergraph $H=(V,E)$ as the sum over all edge sizes, i.e. $|H| = \sum_{e \in E}|e|$.
A hypergraph is \textit{r-uniform} if all edges have size $r$ and it is \textit{d-regular} if all vertices have degree $d$.
Two vertices \(u\) and \(v\) are \textit{adjacent} if both are incident to the same edge, i.e. \mbox{\(E(u) \cap E(v) \neq \emptyset\).}
The \text{neighborhood} \mbox{\(N(v):= \{u \in V \mid \exists e \in E : \{u, v\} \subseteq e\)\}} of a vertex $v$ is the set of vertices that are adjacent to \(v\).
More generally, for a subset of vertices \(S \subseteq V\) the neighborhood is extended to \mbox{\(N(S) = \bigcup_{v \in S} N(v) \setminus S\).}
The same applies for the \textit{closed} neighborhood \mbox{\(N[v] = N(v) \cup \{v\}\)} and its extension \mbox{\(N[S] = N(S) \cup S\)}.
Given a subset \(V' \subset V\), the \textit{subhypergraph} \(H_{V'}\) is defined as \mbox{\(H_{V'} := (V', \{e \cap V' \mid e \in E : e \cap V' \neq \emptyset\})\)}.
If the subhypergraph contains only those hyperedges that are entirely contained within $V'$, i.e. \mbox{\(E' = \{e \in E \mid e \subseteq V'\}\)}, we refer to it as an \textit{induced subhypergraph} and denote it by \(H[V']\).
The removal of a vertex \(v \in V\) is denoted as \mbox{\(H' = H - v\)} instead of explicitly writing \mbox{\(H' = (V \setminus \{v\}, \{e \setminus \{v\} \mid e \in E, e \setminus \{v\} \neq \emptyset\} )\)}.
The same notation applies for subsets \(X \subseteq V\) with \mbox{\(H' = H - X\)} instead of \mbox{\(H' = (V \setminus X, \{e \setminus X \mid e \in E, e \setminus X \neq \emptyset\})\)}.
A \textit{clique} is a set \(C \subseteq V\) such that all vertices are pairwise adjacent. The trivial case of a clique in a hypergraph \mbox{is a hyperedge.}
An \textit{undirected graph} \(G = (V, E)\), with \(n = |V|\) and \(m = |E|\) is a 2-uniform hypergraph with \(E \subseteq \binom{V}{2}\).
A common method to transform an undirected hypergraph \(H = (V, E)\) into an undirected graph \(G = (V, E')\) is the \textit{clique expansion}.
It replaces the hyperedges of the hypergraph with cliques in the graph.
That means for each hyperedge \(e \in E\) and for every pair of vertices \(u,v \in e\), an edge \(\{u,v\}\) is inserted to the edge set \(E'\) of the graph.

A \textit{k-independent set} in a hypergraph is a set \(\mathcal{I} \subseteq V\) such that \mbox{\(|\mathcal{I} \cap e| < k\)} for all \(e \in E\).
In case no parameter $k$ is given it is called a \textit{weak} independent set and no hyperedge is fully contained in \(\mathcal{I}\), i.e. \mbox{\(\mathcal{I} \cap e \neq e\)} for all edges \mbox{\(e \in E\)}.
This work focuses only on 2-independent sets, commonly referred to as \textit{strong} independent sets.
In the context of graphs, there is no distinction between a weak and a strong independent set. 
For simplicity, we refer to strong independent sets as independent sets (IS) throughout this paper.
With IS($H$) we denote the set of all independent sets of \(H\).
An independent set \(\mathcal{I}\) is \textit{maximal} if there exists no other independent set \(\mathcal{I'}\) with \(\mathcal{I} \subsetneq \mathcal{I'}\),
a \textit{maximum independent set} (MIS) \(\mathcal{I} \in \) IS($H$) is an independent set with maximum cardinality.
The size of a maximum independent set is the \textit{independence number} of $H$ and is denoted by \(\alpha(H)\).
The \textit{Maximum Independent Set problem} considered here asks for a maximum independent set in a hypergraph.

A \textit{hitting set} is a subset \(S \subseteq V\) that intersects ("hits") every edge in \(H\),
i.e. \mbox{\(e \cap S \neq \emptyset\)} for all \(e \in E\).
In the context of graphs, this is known as a \textit{vertex cover}.
A \textit{d-hitting set} ensures that the size of every hyperedge $e$ is upper-bounded by a fixed number $d$.
The \textit{Hitting Set problem} asks for a minimum hitting set in a hypergraph.
It is closely related to the MIS problem on hypergraphs when considering \textit{weak} independent sets.
Given a minimum hitting set $S$, its complement \(V \setminus S\) forms a weak MIS in $H$.
This complementary relationship does not hold for the \textit{strong} MIS.
Nevertheless, we show that many reduction rules for the Hitting Set problem can still be applied \mbox{to the MIS problem.}

\section{Related Work}
\label{sec:related_work}

This paper is a summary and extension of the thesis~\cite{ba_wagner}.
We start with covering related work on the extensively studied Maximum Independent Set problem in graphs, while focusing on recent work that incorporates reduction strategies.
Since the Maximum Independent Set problem in graphs is very closely related to the Minimum Vertex Cover problem, we will also cover related work in this area.
Afterwards, we present work on the Independent Set and Hitting Set problem in hypergraphs.

\textbf{Exact Solvers for Graphs.}
Since computing a maximum independent set is NP-hard, the best-known algorithms have worst-case exponential time complexity with respect to the number of vertices and follow the \textit{branch-and-bound} paradigm.
The first non-trivial, theoretical branch-and-bound algorithm, designed by Tarjan and Trojanowski in 1977 \cite{TarjanT77}, runs in \(O(2^\frac{n}{3})\) time and requires polynomial space.
Since then, much research has focused on reducing the base of the exponent~\cite{Bourgeois12,FominGK09,XiaoN17}.
Data reduction rules proved useful in \textit{branch-and-reduce} methods, often reducing the problem size to smaller instances that can be solved exactly.
The branch-and-reduce paradigm applies reduction rules to decrease the input size and uses a branching strategy when no further reduction rules are possible.
Over time, various reduction rules combined with branch-and-bound methods have been developed for the MIS problem and the closely related Vertex Cover problem~\cite{AkibaIwata,Butenko2002,Butenko07,Chen01}.
In 2019, one track in the PACE challenge investigated the Minimum Vertex Cover problem. This track was won by a portfolio approach \weGotYouCovered~\cite{hespe2020wegotyoucovered}.
If not considering portfolio approaches, the best-performing method was \satAndReduce~\cite{plachetta2021sat}, a branch-and-reduce solver that incorporates SAT solving to achieve further speedups.
The current state-of-the-art branch-and-reduce solver for the weighted MIS problem is KaMIS, developed by Lamm~\etal~\cite{KaMIS}.
It is designed to efficiently solve both the weighted and unweighted variants \hbox{of the problem.}
In recent years, more and more work has been done on data reduction rules. A broad overview of how data reduction rules are used for different problems can be found in the survey by Abu-Khzam~\etal~\cite{abu2022recent}. A more problem-specific overview of data reduction rules for the Maximum Weight Independent Set problem can be found in the survey \hbox{by Großmann~\etal~\cite{grossmann2024comprehensive}.}

\textbf{Independent Sets in Hypergraphs.}
For weak independent sets in hypergraphs, a significant amount of research has focused on their theoretical approximability for specific classes of hypergraphs~\cite{AlonAA99,Arras24,BaloghBN21,CohenPST22,HansenL76,HofmeisterL98,QiuW24}.
For the Strong Maximum Independent Set problem in hypergraphs, Halld{\'{o}}rsson and Losievskaja \cite{MSISHeuristic} proposed the two greedy algorithms \textsc{GreedyD} and \textsc{GreedyN}, and established their approximation ratios.
Both algorithms iteratively construct an MIS by selecting either the vertex of minimum degree (\textsc{GreedyD}) or the vertex with the fewest neighbors (\textsc{GreedyN}).

\textbf{$d$-Hitting Sets in Hypergraphs.}
For the $d$-Hitting Set problem, Weihe~\cite{Weihe} first introduced different reduction techniques.
Abu-Khzam \cite{ABUKHZAM2010524} presents a reduction algorithm for the 3-Hitting Set problem and a general reduction for the $d$-Hitting Set problem.
Erd\H{o}s and Rado \cite{SunflowerLemma} introduced in their work about the \(\Delta\)-System Theorem, the concept of sunflowers.
Since then, they have been used for polynomial-time data reduction for the $d$-Hitting Set problem.
Van Bevern \cite{Bevern14} proposes a linear-time reduction algorithm that finds sunflowers to transform a $d$-Hitting Set instance into an equivalent instance with \(O(k^d)\) hyperedges and vertices.
To transfer the problem into the dynamic setting, Bannach et al. \cite{BannachHRT21} introduce \textit{$b$-flowers} as a generalization of sunflowers, since they are easier to identify.

\subfilesbibliography{../bib}

\section{Methods for the MIS Problem on Hypergraphs}
\label{sec:methods}

To solve the reduced hypergraphs, we first introduce an exact approach using an \textit{integer linear program} described in Section~\ref{sec:optimal_solutions}.
Afterwards, we present our new exact data reduction rules (Section~\ref{sec:data_reductions}) for the MIS problem on hypergraphs. These are used in our preprocessing routine to obtain reduced instances.
\subfilesbibliography{../bib}

\subsection{Optimal Solutions}
\label{sec:optimal_solutions}

A \textit{linear program} (LP) is an optimization technique that seeks to maximize or minimize an objective function subject to a system of linear constraints.
An \textit{integer linear program} (ILP) is an LP with the additional constraint that all variables take integer values. 

\textbf{The ILP Formulation.}
One can state the MIS problem on hypergraphs as an integer linear program as follows:
\begin{align}
 \text{maximize} \quad & \sum_{v \in V} x_v \label{eq:objective} \\
 \text{subject to} \quad & \sum_{i \in e} x_i \leq 1, \quad \forall e \in E, \label{eq:constraint1} \\
    & x_v \in \{0, 1\}, \quad \forall v \in V. \label{eq:constraint2}
\end{align}
The decision variables $x_v$ indicate whether a vertex $v$ is included in the MIS \mbox{(\(x_v = 1\))} or excluded \mbox{(\(x_v = 0\)).}
The constraints \eqref{eq:constraint1} ensure the independence property by guaranteeing that, for every edge $e$ in the hypergraph, at most one vertex in $e$ can be included.
Maximizing this objective function corresponds to finding a maximum independent set.

\textbf{The Dual Problem.}
The dual problem to the Maximum Independent Set is the Minimum Edge Cover problem. It can be stated as an integer linear program as follows:

\begin{align}
 \text{minimize} \quad & \sum_{e \in E} y_e \label{eq:dual_objective} \\
 \text{subject to} \quad & \sum_{e: v \in e} y_e \geq 1, \quad \forall v \in V, \label{eq:dual_constraint1} \\
    & y_e \in \{0, 1\}, \quad \forall e \in E. \label{eq:dual_constraint2}
\end{align}

In this formulation, the decision variables $y_e$ indicate wheter an edge is included in the edge cover. The constraint \eqref{eq:dual_constraint1} ensures that every vertex is covered by at least one edge.
As it is the dual problem, a solution to the Minimum Edge Cover problem gives an upper bound on the solution to the Maximum Independent Set problem.

\textbf{Clique Expansion.}
To optimally solve this problem we can also apply the clique expansion to the hypergraph and then utilize existing solvers working on graphs. Given a hypergraph $H=(V,E)$, this transformation works as follows. The set of vertices $V$ stays the same and we start with an empty edge set $E'=\emptyset$. Then, each hyperedge $e\in E$ is expanded into a clique, i.e., for all vertices $u\neq v\in e$ we add an edge $\{u,v\} \cup E'$. An MIS in the resulting graph $G=(V,E')$ corresponds to the strong maximum independent set in $H$~\cite{MSISHeuristic}.

\subsection{Data Reductions}
\label{sec:data_reductions}

The goal of reduction preprocessing is to reduce the hypergraph to a smaller, equivalent instance, allowing the following algorithms to solve the problem more efficiently. This reduction process maintains the optimality of the solutions. That means that if a reduced instance is solved optimally, we can use this solution to construct an optimal solution for the original instance.
Starting from a hypergraph $H$, the presented reduction rules include information on how to construct the reduced hypergraph $H'$. Additionally, the rules contain information on the relationship between the cardinality of an MIS on the reduced hypergraph \(\alpha(H')\) and the cardinality of an MIS on the original hypergraph \(\alpha(H)\), as well as how to lift a solution \(\mathcal{I'}\) from the reduced instance to a solution \(\mathcal{I}\) in the original hypergraph.

In general, there are three ways in which our reduction rules can modify the hypergraph: by including a vertex, excluding a vertex, or removing a hyperedge.
If a vertex $v$ is included in the MIS, all its neighbors are no longer viable, and \(N[v]\) is removed from the hypergraph.
If a reduction rule excludes a vertex $v$, i.e., $v$ cannot be part of any MIS, $v$ is simply removed from the hypergraph.
The removal of vertices can lead to empty hyperedges, which are also removed.
Although edge reductions do not explicitly determine whether a vertex $v$ belongs to the MIS, they significantly enhance the effectiveness of vertex reductions.
Removing hyperedges reduces the structural complexity of the instance and decreases vertex degrees. We start by introducing two edge-reduction rules, followed by vertex-reduction rules.
\begin{reduction} [Size-One Edges] \label{edgeonereduction}
	Let \(e \in E\) with \(|e| = 1\), then we remove \(e\) from $H$. We obtain \(H' = (V, E \setminus \{e\})\), while \(\alpha(H') = \alpha(H)\) and \(\mathcal{I} = \mathcal{I'}\).
\end{reduction}
\begin{proof}
	Size one edges do not affect the adjacency of vertices and can \hbox{therefore be removed.}
\end{proof}

\begin{reduction} [Edge Domination] \label{edgedomreduction}
	Let \mbox{\(e_1, e_2 \in E\)} with \(e_1 \subseteq e_2\), then we remove \(e_1\) from $H$, s.t. \(H' = (V, E \setminus \{e_1\})\), \(\alpha(H) = \alpha(H')\) and \(\mathcal{I} = \mathcal{I'}\).
	We say \(e_2\) \textit{dominates} \(e_1\) and call \(e_1\) the \textit{dominated} edge.
\end{reduction}
\begin{proof}
	Let \mbox{\(e_1\neq e_2 \in E\)} with \(e_1 \subseteq e_2\).
	When removing \(e_1\) from $H$, the adjacencies of any vertex \(v \in e_1\) remain unchanged, as all its neighbors \(u \in e_1\) are still contained in $e_2$.
	This ensures the equivalence of the solutions in $H$ and $H'$ after the \mbox{removal of \(e_1\).}
\end{proof}

The first two vertex reductions are similar to commonly used low-degree reductions for the MIS problem on graphs \cite{AkibaIwata,ChangLZ17,LammSSSW17}.
\begin{reduction} [Degree-Zero] \label{degree-zero}
	Let \(v \in V\) with \(d(v) = 0\), then we can include $v$ into the solution.
	This results in \(H' = H - v\), \(\mathcal{I} = \mathcal{I'} \cup \{v\}\) and \(\alpha(H) = \alpha(H')+1\).
\end{reduction}
\begin{proof}
	Let \(v \in V\) with \(d(v) = 0\).
	Since $v$ has no neighbors in $V$, it is part of \mbox{every MIS in $H$.}
\end{proof}

\begin{reduction} [Degree-One] \label{degree-one}
	Let \(v \in V\) with \(d(v) = 1\), then we can include $v$ into the solution.
	It results in \(H' = H - N[v]\), \(\mathcal{I} = \mathcal{I'} \cup \{v\}\) and \(\alpha(H) = \alpha(H')+1\).
\end{reduction}

\begin{proof}
	Let \(v \in V\) with \(d(v) = 1\) and $E(v) = \{e\}$.
	We can distinguish two cases:
	\begin{enumerate}
		\item There is an MIS \(\tilde{\mathcal{I}}\) that does contain $v$. 
		\item There is an MIS \(\tilde{\mathcal{I}}\) that does not contain $v$.
		      Since $\mathcal{I}$ is maximal, there exists a neighbor \(u \in N(v)\), such that \(u \in \tilde{\mathcal{I}}\).
		      We can construct a new independent set \mbox{\(\mathcal{I^*} = (\tilde{\mathcal{I}} \setminus \{u\}) \cup \{v\}\)} with the same cardinality by including $v$ instead of $u$.
		      Since $v$ is a degree-one neighbor of $u$, replacing $u$ with $v$ does not affect any other vertices in the independent set. 
	\end{enumerate}
	We have shown that in both cases we can find an MIS that contains $v$, therefore, it is safe to include it and we get \(H' = H - N[v]\) and \(\alpha(H) = \alpha(H') + 1\).
\end{proof}

The following reduction rule is based on the concept of twins. Two vertices $u$ and $v$ are \textit{twins} if they are not adjacent and have the same neighborhood.
The difference between twins in a hypergraph and those in a graph is that twins in a graph always have the same degree, whereas twins in a hypergraph can have different degrees and still be considered twins.
The set of twins forms a set of vertices \(T \subset V\), where either all or none of the vertices are included in an MIS.
Figure~\ref{fig:twins} illustrates examples for Reduction~\ref{twinreduction}.

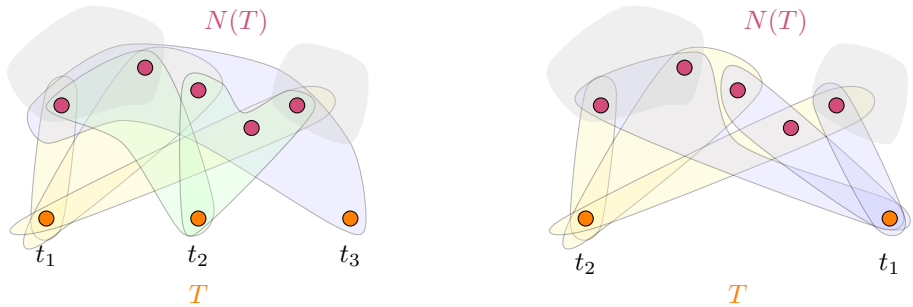
\begin{figure}[b]
    \centering
    \begin{subfigure}{0.49\textwidth} 
        \centering
        \begin{tikzpicture} [scale=1.]
            \filldraw[yellow!40, opacity=0.3, draw=black]
            plot [smooth cycle] coordinates {
                    (-1.8,-0.1) (-2.2,-0.1) (-2,1.7) (-1.6, 1.7)
                };

            \filldraw[yellow!40, opacity=0.3, draw=black]
            plot [smooth cycle] coordinates {
                    (-2, -0.3) (-2.2,0) (-0.8,2.2) (0.3, 1.7)
                };

            \filldraw[yellow!40, opacity=0.3, draw=black]
            plot [smooth cycle] coordinates {
                    (-1.8,-0.1) (-2.3,0) (1.2, 1.7) (1.6, 1.3)
                };

            \filldraw[blue!20, opacity=0.3, draw=black]
            plot [smooth cycle] coordinates {
                    (2.2,-0.1) (1.8,-0.1) (0, 1) (-1, 1.3) (-2,1) (-2.2,1.5) (-1, 2.3) (1, 2) (2, 1)
                };

            \filldraw[green!20, opacity=0.3, draw=black]
            plot [smooth cycle] coordinates {
                    (0.2, 0.2) (-0.2,-0.2) (-1,1) (-2,1.5) (-0.6,2.2)
                };

            \filldraw[green!20, opacity=0.3, draw=black]
            plot [smooth cycle] coordinates {
                    (0.2, 0) (-0.2,-0.1) (-0.2,1.8) (0.3,1.8) (0.6,1.4) (1.2,1.7) (1.5,1.4)
                };

            \filldraw[gray!40, opacity=0.3]
            plot [smooth cycle] coordinates {
                    (-1,1.5) (-2,1.3) (-2.5,2) (-1.5,2.8) (-0.5,2.5) (-0.6,1.6)
                };

            \filldraw[gray!40, opacity=0.3]
            plot [smooth cycle] coordinates {
                    (2,1) (1,1.5) (1.3,2.3) (2.2,2)
                };

            \node[draw, circle, fill=orange, inner sep=2pt] (v1) at (-2,0) {};
            \node[draw, circle, fill=orange, inner sep=2pt] (v2) at (2,0) {};
            \node[draw, circle, fill=orange, inner sep=2pt] (v8) at (0,0) {};

            \node[draw, circle, fill=purple!70, inner sep=2pt] (v3) at (-1.8,1.5) {};
            \node[draw, circle, fill=purple!70, inner sep=2pt] (v4) at (-0.7,2) {};
            \node[draw, circle, fill=purple!70, inner sep=2pt] (v5) at (0,1.7) {};
            \node[draw, circle, fill=purple!70, inner sep=2pt] (v6) at (0.7,1.2) {};
            \node[draw, circle, fill=purple!70, inner sep=2pt] (v7) at (1.3,1.5) {};

            \node at (-2, -0.5) {$t_1$};
            \node at (0,-0.5) {$t_2$};
            \node at (2, -0.5) {$t_3$};

            \node[purple!70] at (0.5,2.6) {$N(T)$};
            \node[orange] at (0,-1) {$T$};
        \end{tikzpicture}
        \label{subfig1:twins}
    \end{subfigure}
    \hfill
    \begin{subfigure}{0.49\textwidth} 
        \centering
        \begin{tikzpicture} [scale=1.]
            \filldraw[yellow!40, opacity=0.3, draw=black]
            plot [smooth cycle] coordinates {
                    (-1.8,-0.1) (-2.2,-0.1) (-2,1.7) (-1.6, 1.7)
                };

            \filldraw[yellow!40, opacity=0.3, draw=black]
            plot [smooth cycle] coordinates {
                    (-2, -0.3) (-2.2,0) (-0.8,2.2) (0.3, 1.7)
                };

            \filldraw[yellow!40, opacity=0.3, draw=black]
            plot [smooth cycle] coordinates {
                    (-1.8,-0.1) (-2.3,0) (1.2, 1.7) (1.6, 1.3)
                };

            \filldraw[blue!20, opacity=0.3, draw=black]
            plot [smooth cycle] coordinates {
                    (2.1,0.2) (1.8,-0.1) (-2.1, 1.5) (-0.5, 2.2) (0,1)
                };

            \filldraw[blue!20, opacity=0.3, draw=black]
            plot [smooth cycle] coordinates {
                    (2.1,0.2) (1.8,-0.1) (0, 1.2) (-0.2, 1.9) (0.2,1.9)
                };

            \filldraw[blue!20, opacity=0.3, draw=black]
            plot [smooth cycle] coordinates {
                    (2.2,0) (1.8,-0.1) (1, 1.5) (1.4, 1.7)
                };

            \filldraw[gray!40, opacity=0.3]
            plot [smooth cycle] coordinates {
                    (-1,1.5) (-2,1.3) (-2.5,2) (-1.5,2.8) (-0.5,2.5) (-0.6,1.6)
                };

            \filldraw[gray!40, opacity=0.3]
            plot [smooth cycle] coordinates {
                    (2,1) (1,1.5) (1.3,2.3) (2.2,2)
                };

            \node[draw, circle, fill=orange, inner sep=2pt] (v1) at (-2,0) {};
            \node[draw, circle, fill=orange, inner sep=2pt] (v2) at (2,0) {};

            \node[draw, circle, fill=purple!70, inner sep=2pt] (v3) at (-1.8,1.5) {};
            \node[draw, circle, fill=purple!70, inner sep=2pt] (v4) at (-0.7,2) {};
            \node[draw, circle, fill=purple!70, inner sep=2pt] (v5) at (0,1.7) {};
            \node[draw, circle, fill=purple!70, inner sep=2pt] (v6) at (0.7,1.2) {};
            \node[draw, circle, fill=purple!70, inner sep=2pt] (v7) at (1.3,1.5) {};

            \node at (2,-0.6) {$t_1$};
            \node at (-2,-0.6) {$t_2$};

            \node[purple!70] at (0.5,2.6) {$N(T)$};
            \node[orange] at (0,-1) {$T$};
        \end{tikzpicture}
        \label{subfig2:twins}
    \end{subfigure}
    \caption{Left: A twin set $T = \{t_1,t_2,t_3\}$ with $\Delta_T = 1$, s.t. $|T| > \Delta_T$. Thus, all vertices in $T$ can be added to the solution $\mathcal{I}$ and $N[T]$ can be removed. Right: The twin set \(T = \{t_1,t_2\}\) is of size $2$, but \(\Delta_T = 3\).  In this case, we do not add the vertices in \(T\) to the solution.}
    \label{fig:twins}
\end{figure}
\begin{reduction} [Twins] \label{twinreduction}
	Let \(T \subset V\) be a set of twins, i.e. for all \(s\neq t \in T\) holds \(N(s) = N(t)\) and s,t are not adjacent.
	Let \(\delta_T := min_{t \in T}\) $d(t)$ be the minimum degree of all twins in $T$.
	If \(|T| \geq \delta_T\), all vertices of $T$ are part of an MIS in $H$.
	We obtain \(H' = H - N[T]\), with \(\mathcal{I} = \mathcal{I'} \cup T\) and \(\alpha(H) = \alpha(H') + |T|\).
\end{reduction}

\begin{proof}
	Let \(T \subset V\) be a set of twins and $t\in T$ the twin with minimum degree $\delta_T$.
	For the subhypergraph \(H_{N[T]}\), it holds that either $T$ or some vertices from $N(T)$ are in the MIS. If a maximum independent set in $H_{N(T)}$ is smaller than or equal to $|T|$, we can include all twins in the MIS and exclude their neighbors.
	Therefore, we show that the independence number of the subhypergraph \(H_{N(T)}\) can be bounded by \(\delta_T\).

	Recall that \(H_{N(T)}\) contains all vertices in $N(T)$ and all edges incident to the vertices in $T$ while not containing the twins themselves.
	By definition, $t$ is adjacent to all of \(N(T)\), and therefore the edges incident to $t$, i.e. $E(t)$, represent a set of edges that covers $N(T)$.
	Since the Minimum Edge Cover problem is the dual to the Maximum Independent Set, this gives us an upper bound on the solution size, and we get \(\alpha(H_{N(T)}) \leq |E(t)| =  \delta_T\).

	If \(|T| \geq \delta_T\), the number of twins that can be added to the IS exceeds the number of vertices in \(N(T)\) that could possibly be included, as we have shown that \mbox{\(\alpha(H_{N(T)}) \leq \delta_T\).}
	Therefore, we include \(T\) in the IS and obtain \mbox{\(H' = H - N[T]\).}

	If \(|T| < \delta_T\), \(\alpha(H_{N(T)}) \leq |T|\) is \textit{not} assured and it might be possible that the independent set in \(H_{N(T)}\) is larger than \(|T|\).
	In this case, we do not include the vertices of \(T\).
\end{proof}

The next reduction focuses on detecting \textit{sunflower} structures in the hypergraph.
These structures originate from Erd\H{o}s and Rado~\cite{SunflowerLemma} and have been used in reduction algorithms and analysis of the $d$-Hitting Set problem~\cite{Bevern14}.
Figure \ref{fig:sunflowerreduction} illustrates an example of an \hbox{application of Reduction~\ref{sunflowerreduction}.}
\begin{definition} [Sunflower] \label{def:sunflower}
	Distinct hyperedges \(S_1,\dots,S_k\), form a $k$-sunflower if they all share the same intersection, i.e.
	$S_i \cap S_j = \bigcap\limits_{t=1}^{k}S_t$, for all $i \neq j$.
	We call \(C = \bigcap\limits_{t=1}^{k}S_t\) the core and \(S_i \setminus C\) for \(i = 1,\dots,k\) the petals of the sunflower.
\end{definition}

    \begin{figure}[b]
        \centering
        \begin{subfigure}{0.4\textwidth}
            \centering
            \begin{tikzpicture}[scale=0.5]
                \filldraw[gray!50, opacity=0.3] (-1.5,0) ellipse (3 and 1.5);
                \draw[gray!60, very thick] (-1.5,0) ellipse (3 and 1.5);
                \node at (-5,0) {\(e_1\)};
    
                \filldraw[gray!50, opacity=0.3] (0,1.5) ellipse (1.5 and 3);
                \draw[gray!60, very thick] (0,1.5) ellipse (1.5 and 3);
                \node at (0,5) {\(e_2\)};
    
                \filldraw[gray!50, opacity=0.3] (1.5,0) ellipse (3 and 1.5);
                \draw[gray!60, very thick] (1.5,0) ellipse (3 and 1.5);
                \node at (5,0) {\(e_3\)};
    
                \begin{scope}
                    \clip (-1.5,0) ellipse (3 and 1.5);
                    \clip (0,1.5) ellipse (1.5 and 3);
                    \clip (1.5,0) ellipse (3 and 1.5);
                    \fill[red!50, opacity=0.5] (-2,-2) rectangle (4,4);
                \end{scope}
    
                \node[draw, circle, fill=black, inner sep=2pt] (v1) at (0,-0.5) {};
                \node[draw, circle, fill=black, inner sep=2pt] (v2) at (0.5,0.4) {};
                \node[draw, circle, fill=black, inner sep=2pt] (v3) at (-0.5,0.4) {};
    
                \node[circle, fill=gray!60, inner sep=2pt] (v6) at (-3,0.5) {};
                \node[circle, fill=gray!60, inner sep=2pt] (v7) at (-3.5,-0.5) {};
                \node[circle, fill=gray!60, inner sep=2pt] (v8) at (0.2,4) {};
                \node[circle, fill=gray!60, inner sep=2pt] (v9) at (-0.4,3.2) {};
                \node[circle, fill=gray!60, inner sep=2pt] (v10) at (0.4,2.5) {};
                \node[circle, fill=gray!60, inner sep=2pt] (v11) at (3.5, 0.4) {};
    
                \node[above] at (v1) {$v_1$};
                \node[above] at (v2) {$v_2$};
                \node[above] at (v3) {$v_3$};
            \end{tikzpicture}
        \end{subfigure}
        \hspace{1cm}
        \begin{subfigure}{0.4\textwidth}
            \centering
            \begin{tikzpicture}[scale=0.5]
                \filldraw[gray!50, opacity=0.3] (-1.5,0) ellipse (3 and 1.5);
                \draw[gray!60, very thick] (-1.5,0) ellipse (3 and 1.5);
                \node at (-5,0) {\(e_1\)};
    
                \filldraw[gray!50, opacity=0.3] (0,1.5) ellipse (1.5 and 3);
                \draw[gray!60, very thick] (0,1.5) ellipse (1.5 and 3);
                \node at (0,5) {\(e_2\)};
    
                \filldraw[gray!50, opacity=0.3] (1.5,0) ellipse (3 and 1.5);
                \draw[gray!60, very thick] (1.5,0) ellipse (3 and 1.5);
                \node at (5,0) {\(e_3\)};
    
                \begin{scope}
                    \clip (-1.5,0) ellipse (3 and 1.5);
                    \clip (0,1.5) ellipse (1.5 and 3);
                    \clip (1.5,0) ellipse (3 and 1.5);
                    \fill[yellow!50, opacity=0.5] (-2,-2) rectangle (4,4);
                \end{scope}
    
                \node[draw, circle, fill=black, inner sep=2pt] (v1) at (0,0) {};
    
                \node[circle, fill=gray!60, inner sep=2pt] (v6) at (-3,0.5) {};
                \node[circle, fill=gray!60, inner sep=2pt] (v7) at (-3.5,-0.5) {};
                \node[circle, fill=gray!60, inner sep=2pt] (v8) at (0.2,4) {};
                \node[circle, fill=gray!60, inner sep=2pt] (v9) at (-0.4,3.2) {};
                \node[circle, fill=gray!60, inner sep=2pt] (v10) at (0.4,2.5) {};
                \node[circle, fill=gray!60, inner sep=2pt] (v11) at (3.5, 0.4) {};
    
                \node[above] at (v1) {$v_1$};
            \end{tikzpicture}
        \end{subfigure}
        \caption{A $3$-sunflower before (left) and after (right) the sunflower reduction.}
        \label{fig:sunflowerreduction}
    \end{figure}
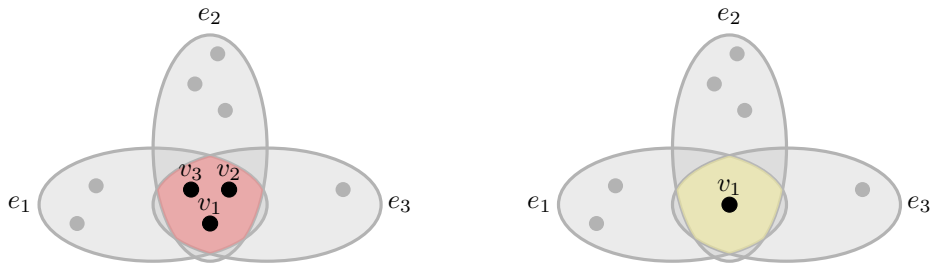

\begin{reduction} [Sunflower] \label{sunflowerreduction}
	Let \(S_1,\dots,S_k \in E\) be a $k$-sunflower in $H$, \(C\) be its core such that all vertices $v\in C$ share the same incident edges. Let $v\in V$ be an arbitrary core vertex, i.e. \(v \in C\).
	Then, we can exclude the vertices \(u \in C \setminus \{v\}\). This results in \(H' = H - (C \setminus \{v\})\), \(\mathcal{I} = \mathcal{I'}\) and \(\alpha(H) = \alpha(H')\).
\end{reduction}
\begin{proof}
	Since all vertices in the core $C$ of a sunflower share the same set of incident edges, they are adjacent to each other.
	Consequently, at most one vertex \(v \in C\) can be included in an IS.
	This vertex $v$ can be chosen arbitrarily, as the identical incidence structure ensures that all vertices in $C$ have the same neighborhood.
	Therefore, we exclude all vertices in the core of the sunflower except for $v$.
\end{proof}

\definecolor{mygreen}{cmyk}{0.8, 0, 0.8, 0}  
    \begin{figure}[t]
        \centering
    
        \begin{subfigure}{0.4\textwidth}
            \centering
            \begin{tikzpicture}[scale=0.75]
                \filldraw[rotate around={57.5:(-0.75,1.5)}, mygreen!70, opacity=0.3] (-0.75,1.5) ellipse (2.5 and 1);
                \draw[rotate around={57.5:(-0.75,1.5)}, mygreen!70, very thick] (-0.75,1.5) ellipse (2.5 and 1);
                \node[mygreen] at (-2.1,2.1) {\(e\)};
        
                \filldraw[rotate around={-57.5:(0.75,1.5)}, blue!50, opacity=0.3] (0.75,1.5) ellipse (2.5 and 1);
                \draw[rotate around={-57.5:(0.75,1.5)}, blue!60, very thick] (0.75,1.5) ellipse (2.5 and 1);
                \node[blue] at (2.1,2.1) {\(f\)};
        
                \filldraw[purple!50, opacity=0.3] (0,0.15) ellipse (2.5 and 1);
                \draw[purple!60, very thick] (0,0.15) ellipse (2.5 and 1);
                \node[purple] at (0,-1.2) {\(g\)};
                
                \filldraw[gray!40, opacity=0.3] 
                plot [smooth cycle] coordinates {
                    (1.8,-1.8) (2.5,-0.2) (-1,0) (-2,-0.3) (-2.2,-1.2) (-1.2,-2)
                };
                \filldraw[gray!40, opacity=0.3] 
                plot [smooth cycle] coordinates {
                    (3, -1) (2, -0.6) (1.5, 0) (2.5, 0.8) (3.5, 0.3) (3.3,-0.4)
                };
                \filldraw[gray!40, opacity=0.3] 
                plot [smooth cycle] coordinates {
                    (-1.2, 0.8) (-1,0) (-1.2, -1.5) (-2.5,-1.4) (-3.2, -0.5) (-2.5, 0.8)
                };
        
                \node[draw, circle, fill=blue, inner sep=2pt] (v1) at (0,3) {};
                \node[draw, circle, fill=mygreen, inner sep=2pt] (v5) at (-2,0) {};
                \node[draw, circle, fill=mygreen, inner sep=2pt] (v6) at (-1.5,0.3) {};
                \node[draw, circle, fill=mygreen, inner sep=2pt] (v7) at (-1.5,-0.3) {};
                \node[draw, circle, fill=purple, inner sep=2pt] (v8) at (2,0) {};
                \node[draw, circle, fill=purple, inner sep=2pt] (v9) at (1.5,0.3) {};
                \node[draw, circle, fill=purple, inner sep=2pt] (v10) at (1.5,-0.3) {};
        
                \node[circle, fill=gray!50, inner sep=2pt] (v14) at (-0.6,-0.5) {};
                \node[circle, fill=gray!50, inner sep=2pt] (v15) at (0.2,-0.3) {};
        
                \node[circle, fill=gray!40, inner sep=2pt] (v16) at (-3,-0.3) {};
                \node[circle, fill=gray!40, inner sep=2pt] (v17) at (-2.5,-0.8) {};
                \node[circle, fill=gray!40, inner sep=2pt] (v18) at (3,0) {};

                \node[circle, fill=gray!40, inner sep=2pt] (v20) at (-0.8,-1.4) {};
                \node[circle, fill=gray!40, inner sep=2pt] (v21) at (-1.6,-1.3) {};
                \node[circle, fill=gray!40, inner sep=2pt] (v22) at (1.7,-1) {};
                \node[circle, fill=gray!40, inner sep=2pt] (v23) at (0.9,-1.7) {};
        
                \node[above] at (v1) {$a_1$};
                \node[above] at (v5) {$b_1$};
                \node[above] at (v6) {$b_2$};
                \node[above] at (v7) {$b_3$};
                \node[above] at (v8) {$c_1$};
                \node[above] at (v9) {$c_2$};
                \node[above] at (v10) {$c_3$};
            \end{tikzpicture}
        \end{subfigure}
        \hspace{1cm}
        \begin{subfigure}{0.4\textwidth}
            \centering
            \begin{tikzpicture}[scale=0.75]
            \filldraw[rotate around={57.5:(-0.75,1.5)}, gray!50, opacity=0.3] (-0.75,1.5) ellipse (2.5 and 1);
            \draw[rotate around={57.5:(-0.75,1.5)}, gray!60, very thick] (-0.75,1.5) ellipse (2.5 and 1);
    
            \filldraw[rotate around={-57.5:(0.75,1.5)}, gray!50, opacity=0.3] (0.75,1.5) ellipse (2.5 and 1);
            \draw[rotate around={-57.5:(0.75,1.5)}, gray!60, very thick] (0.75,1.5) ellipse (2.5 and 1);
    
            \filldraw[yellow!40, opacity=0.3, draw=black] (-0.2,-0.4) ellipse (0.75 and 0.35);
            
            \filldraw[yellow!40, opacity=0.3, draw=black] 
            plot [smooth cycle] coordinates {
                (1.8,-1.8) (2,-1) (1.2,-0.6) (0,0) (-1.2,-0.6) (-2,-1.2) (-1.2,-2)
            };
            \filldraw[yellow!40, opacity=0.3, draw=black] 
            plot [smooth cycle] coordinates {
                (3, -0.6) (2.6, -0.1) (2.7, 0.4) (3.5, 0.3) (3.3,-0.4)
            };
            \filldraw[yellow!40, opacity=0.3, draw=black] 
            plot [smooth cycle] coordinates {
                (-1.3, -1.5) (-2.5,-1.4) (-3.2, -0.5) (-2.9, 0) (-2.3, -0.7) (-1.5, -0.9)
            };
    
            \node[draw, circle, fill=red, inner sep=2pt] (v1) at (0,3) {};
            \node[draw, circle, fill=gray, inner sep=2pt] (v5) at (-2,0) {};
            \node[draw, circle, fill=gray, inner sep=2pt] (v6) at (-1.5,0.3) {};
            \node[draw, circle, fill=gray, inner sep=2pt] (v7) at (-1.5,-0.3) {};
            \node[draw, circle, fill=gray, inner sep=2pt] (v8) at (2,0) {};
            \node[draw, circle, fill=gray, inner sep=2pt] (v9) at (1.5,0.3) {};
            \node[draw, circle, fill=gray, inner sep=2pt] (v10) at (1.5,-0.3) {};
    
            \node[draw, circle, fill=yellow!50, inner sep=2pt] (v14) at (-0.6,-0.5) {};
            \node[draw, circle, fill=yellow!50, inner sep=2pt] (v15) at (0.2,-0.3) {};
    
            \node[draw, circle, fill=yellow!50, inner sep=2pt] (v16) at (-3,-0.3) {};
            \node[draw, circle, fill=yellow!50, inner sep=2pt] (v17) at (-2.5,-0.8) {};
            \node[draw, circle, fill=yellow!50, inner sep=2pt] (v18) at (3,0) {};
    
            \node[draw, circle, fill=yellow!50, inner sep=2pt] (v20) at (-0.8,-1.4) {};
            \node[draw, circle, fill=yellow!50, inner sep=2pt] (v21) at (-1.6,-1.3) {};
            \node[draw, circle, fill=yellow!50, inner sep=2pt] (v22) at (1.7,-1) {};
            \node[draw, circle, fill=yellow!50, inner sep=2pt] (v23) at (0.9,-1.7) {};
    
            \node[above] at (v1) {$a_1$};
            \node[above] at (v5) {$b_1$};
            \node[above] at (v6) {$b_2$};
            \node[above] at (v7) {$b_3$};
            \node[above] at (v8) {$c_1$};
            \node[above] at (v9) {$c_2$};
            \node[above] at (v10) {$c_3$};
            \end{tikzpicture}
        \end{subfigure}
        \caption{    
        An example of Reduction~\ref{isolatedcliquereduction} with the original instance on the left, reduced to the instance on the right.
        }
        \label{fig:triangle}
    \end{figure}
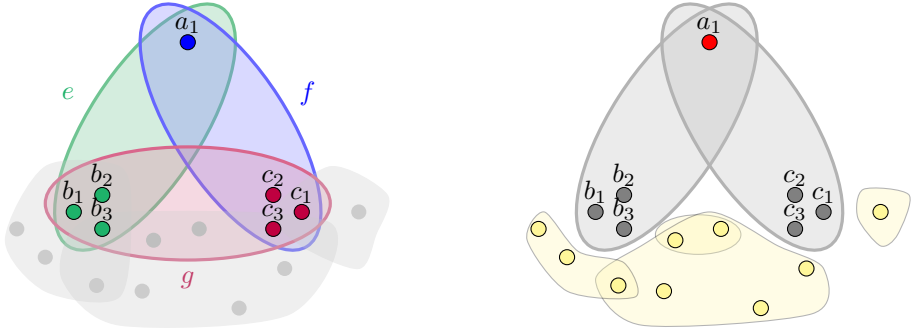 

Next, we want to generalize the \textit{simplicial vertex reduction}, which has proven highly effective as an MIS reduction for many real-world graph instances, as shown by Strash~\cite{Strash16}.
The key idea of this reduction is that in a clique, at most one vertex can belong to an independent set.
While in the context of graphs, the size of a clique \(C\) is typically defined as the number of vertices it contains, we define the clique size \(|C|\) for hypergraphs as the number of edges in the induced subhypergraph \(H[C]\).

\begin{definition} [Simplicial Vertex]
	A simplicial vertex is a vertex $v$ whose neighborhood forms a clique. Specifically, there exists a clique \(C \subseteq V\) such that \(C \cap N(v) = N(v)\).
\end{definition}
Vertices with this property can be reduced according to the following \mbox{reduction rule.}
Figure \ref{fig:triangle} shows an application of Reduction~\ref{isolatedcliquereduction}.
\begin{reduction} [Simplicial Vertex] \label{isolatedcliquereduction}
	Let \(C \subseteq V\) be a clique and let $v$ be a simplicial vertex with \(N(v) = C\), then we can include $v$.
	This results in \(H' = H - (C \cup \{v\}\)), \(\mathcal{I} = \mathcal{I'} \cup \{v\}\) and \(\alpha(H) = \alpha(H') +1\).
\end{reduction}
\begin{proof}
	Let $C\subseteq V$ be a clique in $H$ and $A$ be the set of simplicial vertices of $C$.
	Only one vertex of a clique can be contained in an independent set \(\tilde{\mathcal{I}}\) of $H$, because all vertices of the clique are adjacent to each other.
	Assume towards a contradiction that no vertex $v \in A$ belongs to \(\tilde{\mathcal{I}}\).
	We have two cases:
	\begin{enumerate}
		\item A vertex $u\in C\setminus A$, is included in the independent set.
		      Then, we can construct a new independent set \(\mathcal{I^*} = (\tilde{\mathcal{I}} \setminus \{u\}) \cup \{v\}\) that includes $v$ instead of $u$ with \hbox{the same cardinality.}
		\item No vertex of the clique belongs to \(\tilde{\mathcal{I}}\).
		      Then, a vertex $v \in A$ can safely be added to the independent set, since it has no neighbors outside of the clique.
		      We can choose $v \in A$ arbitrarily because of the identical incidence structure of all vertices in $A$.
		      By including $v$ in \(\tilde{\mathcal{I}}\), itself and all its neighbors, i.e., the vertices of the clique, are removed from the hypergraph, and it holds \(\alpha(H) = \alpha(H') +1\).
	\end{enumerate}
\end{proof}

Although this reduction theoretically applies to cliques of any size, verifying whether the neighborhood of a vertex $v$ forms a clique takes \(O(|N(v)|^2 \cdot |C|)\) time,
making the reduction more computationally expensive as the clique size grows and the degrees of its vertices increase.
For this reason, our reduction algorithm only checks for cliques of size 3.

The next reduction we introduce is Reduction~\ref{vertexdomination}. It is based on the vertex domination concept introduced by Weihe \cite{Weihe} and later used by Niedermeier and Rossmanith~\cite{NIEDERMEIER200389} as one of two preprocessing rules in a reduction algorithm for the 3-Hitting Set problem.
For the MIS problem, we apply a variant of this reduction, similar to the one described by Chang et al.~\cite{ChangLZ17} for graphs.
For two adjacent vertices $u,v\in V$ we say $u$ \textit{dominates} $v$ \hbox{if \(N[v] \subseteq N[u]\).}
\begin{reduction} [Vertex Domination] \label{vertexdomination}
	Let $u,v\in V$ be two adjacent vertices such that $u$ dominates $v$. Then, we can exclude $u$, resulting in \(H' = H - u\), $\mathcal{I}' = \mathcal{I}$ and \(\alpha(H) = \alpha(H')\).
\end{reduction}
\begin{proof}
	Let \(u,v\in V\) with \(N[v]\subseteq N[u]\) and \(\tilde{\mathcal{I}}\) be an MIS of $H$. If $u$ is not in \(\tilde{\mathcal{I}}\), we are done.
	Otherwise, suppose \(u \in \tilde{\mathcal{I}}\).
	Then, we can construct a new MIS \(\mathcal{I^*} = (\mathcal{I} \setminus \{u\}) \cup \{v\}\) with the same cardinality, since the domination property ensures that any vertex adjacent to $v$ is also adjacent to $u$.
	Thus, \(\mathcal{I^*}\) is an MIS of \mbox{$H$ excluding $u$.}
\end{proof}

The following \textit{unconfined reduction} is not restricted to a local region but can potentially extend to the entire hypergraph.
This reduction was first proposed for graphs by Xiao and Nagamochi in \cite{XiaoN13}.
The key idea is to remove a vertex $v$ if the assumption that \emph{every} maximum independent set of a graph $G$ includes $v$ leads to a contradiction.
This principle extends to the hypergraph context in the same way as it does in the graph context.
Before we explain the reduction, we define the concept of a \textit{child} in the hypergraph, which is an essential part of the reduction.

\begin{definition} [Child]
	Given an independent set $S$ of $H$, a vertex \(u \in N(S)\) is a \textit{child} of $S$ if it has exactly one neighbor \(s \in S\), i.e. \(|N(u) \cap S| = 1\). This neighbor $s$ is called the \hbox{parent of $u$.}
\end{definition}
The concept used in the unconfined reduction rule is based on the following observation described in Lemma~\ref{lem:unconfined}.
\begin{lemma}\label{lem:unconfined}
	Let $S$ be an independent set that is contained in all maximum independent sets of $H$. Then, any maximum independent set has to contain one vertex $v \in N(u) \setminus N[S]$ for each child $u\in N(S)$.
\end{lemma}
\begin{proof}
	Let $S$ be an independent set of $H$ that is contained in all MIS of $H$.
	We assume towards a contradiction that there is a maximum independent set $I$ of $H$ and some child $u\in N(S)$, such that $I \cap (N(u)\setminus N[S])=\emptyset$. Then, the parent $u'\in S\cap N(u)$ of $u$ belongs to $I$, since $S$ is contained in all MIS. Now, we can construct a new independent set $I'=I\setminus \{u'\}\cup \{u\}$. This set is independent in $H$ since $N(u)\cap (S\setminus \{u'\})=\emptyset$ and it holds $I|=|I'|$. By that, we have constructed a new MIS $I'$ that does not contain $S$ entirely, which leads \hbox{to the contradiction.}
\end{proof}

\begin{reduction} [Unconfined] \label{unconfinedreduction}
	To determine whether a vertex $v$ is unconfined, we initialize \(S = \{v\}\) and apply the following simple algorithm:
	\begin{enumerate}
		\item Find \(u \in N(S)\) such that $u$ is a child of $S$. If there exists no such vertex, \mbox{$v$ is confined.}
		\item One of the conditions holds:
		      \begin{enumerate}
			      \item \label{item:unconfined-a} If \(N(u) \setminus N[S] = \emptyset\), $v$ is unconfined.
			      \item \label{item:unconfined-b} If \(|N(u) \setminus N[S]| > 1\), $v$ is confined.
			      \item \label{item:unconfined-c} If \(N(u) \setminus N[S]\) is a single vertex $w$, add $w$ to S and repeat the \mbox{algorithm from 1.}
		      \end{enumerate}
	\end{enumerate}
	If the algorithm terminates at \ref{item:unconfined-a}, the assumption that every maximum independent set in $H$ contains $v$ is false, and $v$ is unconfined.
	Any unconfined vertex can be excluded resulting in \(H' = H - v\), $\mathcal{I}' = \mathcal{I}$ and \(\alpha(H) = \alpha(H')\).
	If the algorithm terminates at step \ref{item:unconfined-b}, the set $S$ confines the vertex $v$, meaning $v$ can not be excluded.
	The algorithm then proceeds with a new vertex \(u \in V\), \mbox{initializing \(S = \{u\}\).}
\end{reduction}

Reduction~\ref{unconfinedreduction} repeatedly grows the set $S$ according to Lemma~\ref{lem:unconfined}. If the algorithm reaches \ref{item:unconfined-a}, we have proven that there is an MIS that does not contain $v$, and therefore we can remove $v$ from the hypergraph.

\textbf{The Reduction Framework.}
The application of our reductions follows a predefined order, and whenever a reduction reduces the hypergraph, the algorithm restarts with the first reduction.
This iterative progress continues until no further reduction can be applied.
To avoid exhaustively applying reductions to every vertex and edge of the hypergraph in each iteration, we only test vertices and edges for reduction again if their neighborhood has changed.
Our preprocessing uses the following set of rules and the given order
	[\ref{degree-zero}, \ref{degree-one}, \ref{twinreduction}, \ref{sunflowerreduction}, \ref{isolatedcliquereduction}, \ref{vertexdomination}, \ref{unconfinedreduction}, \ref{edgedomreduction}].
Our reductions are mostly ordered by increasing complexity. We do not consider other orderings of data reductions, since for other problems reordering has no meaningful impact on performance as Großmann~\etal~\cite{GrossmannLSS24} demonstrated and preliminary experiments for our problem confirm this.
Reduction~\ref{edgeonereduction} is not listed here, since it is checked in other reduction rules and does not enable further reduction progress.
The overall most time consuming reduction in this set is Reduction~\ref{unconfinedreduction}. When omitting Reduction~\ref{unconfinedreduction}, the reduction preprocessing can be slightly faster, but leads to larger reduced instances.

\subfilesbibliography{../bib}

\section{Experimental Evaluation}
\label{sec:experimental_evaluation}

In this section, we demonstrate the effectiveness of our data reduction techniques and compare different methods for solving the MIS problem in hypergraphs using these reductions.
All algorithms and data structures are implemented in C++17 and compiled with g++ (gcc) 9.4, with compiler optimizations enabled (-O3).
The experiments are executed on a machine equipped with an AMD EPYC 9754 128-Core CPU running at 2.25GHz, 256MB L3 cache, and 768\,GB of main memory.
We use Gurobi version 12.0.3~\cite{gurobi}, a commercial ILP solver, to solve the ILP formulation on the hypergraph \ilp{} and on the graph \ilpg{}.
We run each configuration with a one-hour time limit and use four different seeds, taking the geometric mean time as the result. Every solver is restricted to a single thread to ensure fair comparisons.
Whenever our preprocessing is added, we mark it with a \strong{} prefix to the solver name.

We use performance profiles~\cite{PP} for comparing running times across different methods and instances.
In the performance profiles, we plot the fraction of instances as a function of the performance ratio \(\tau\).
Here, \(\tau\) is a factor telling us how much slower the running time of a method is compared to the fastest approach on the same instance.
In this context, \(\tau\) is greater than or equal to 1.
A value of \(\tau=2\), for example, means that the method took twice as long as \hbox{the fastest approach.}

\textbf{Instances.}
For all experiments, we use a randomly selected subset of the \(M_{HG}\) hypergraph dataset, provided by Gottesbüren~\etal~\cite{GottesburenHMSS24} that consists of 488 instances.
Our subset includes a total of 50 instances presented in Table~\ref{tab:reductions_overview} in the appendix. In detail, we use two instances for circuit design (ispd98~\cite{Ispd}), one instance from the DAC 2012 Routability-Driven Placement Contest (dac2012~\cite{DAC12}), and 21 instances derived from general matrices in the SuiteSparse Matrix Collection (ssmc~\cite{SPM11}). Furthermore, we included 26 instances derived from the International SAT Competition 2014 (sat14~\cite{SAT14}).
For the solver comparison, we restrict the set to 32 instances that are optimally solvable by at least one \hbox{of the algorithms.}

\subsection{Reduction Impact on Instance Size}
\label{sec:experiments_instance_size}

\begin{figure}[b]
    \begin{minipage}{.47\textwidth}
        \centering
        \begin{tikzpicture}
            \begin{axis}[height=3cm, width=0.85\textwidth, legend to name=legend_reduction, ylabel=$|H|$, xlabel=$t$,xmode=log,grid,xmax=120, ymax=1,ylabel style={at={(axis description cs:0.,1.05)}, anchor=south, rotate=-90}]
                \addplot+[only marks,opacity=0.6] table[x index=0,y index=1,col sep=comma] {data/figure_data/pred_data.csv};
            \end{axis}
        \end{tikzpicture}
        \captionof{figure}{For each instance, we plot the remaining hypergraph size $|H|$ (in percent) and its reduction time $t$ (in seconds).}\label{fig:reduction_effect}
    \end{minipage}\hfill\begin{minipage}{.47\textwidth}
        \centering
\begin{tikzpicture}
    \begin{axis}[
            ybar,
            bar width=15pt,
            ymajorgrids,
            width=.9\textwidth,
            height=2cm,
            ymax=4.2,
            xtick={1,2,3,4,5},
            enlarge x limits=0.1,
            ytick={1,2,3,4},
            y axis line style={-Stealth},
            x axis line style={-},
            yticklabels={1,2,3,4},
            xticklabels={\ilp,\ilpg,\satAndReduce,\struction,\weGotYouCovered},
            ylabel={Speedup}, 
            xticklabel style={rotate=20, anchor=east,xshift=4pt,yshift=-2pt},
            xlabel={},
            ylabel style={at={(axis description cs:0.,1.05)}, anchor=south, rotate=-90}
        ]
        \addplot[Dark2-B, fill=Dark2-B!70, postaction={ pattern=dots }] coordinates {
                (1,1.96)
                (2,3.84)
                (3,2.45)
                (4,1.16)
                (5,1.09)
            };

    \end{axis}
\end{tikzpicture}
\captionof{figure}{Solver-wise speedup achieved when using reductions. We present the geometric mean speedup over all instances solvable \hbox{by all algorithms.}}\label{fig:speedup_solver_comparison}
    \end{minipage}
\end{figure}
\begin{table}
    \caption{We present the average instance properties and the reduction effect. We show the mean number of vertices $n$, the mean number of edges $m$, average edge size $|e|$, average hypergraph size $|H|$, and the mean reduction time $t$ in seconds.}\label{tab:reduction_effect}
    \centering
    \begin{tabular}{lrrrrr}
                  & $n$               & $m$               & $|e|$            & $|H|$                       & $t$[s]          \\
        \toprule
        original  & \numprint{315399} & \numprint{676599} & \numprint{18.18} & \numprint{12.3}$\times10^6$ & 0               \\
        \strong{} & \numprint{195602} & \numprint{302658} & \numprint{9.08}  & \numprint{2.7}$\times10^6$  & \numprint{6.76}
    \end{tabular}
\end{table}

In the first part of our experimental evaluation, we examine the reduction effect on the instance size. In Table~\ref{tab:reductions_overview} in the Appendix, we present detailed per-instance information on our dataset and the reduction effect. In Table~\ref{tab:reduction_effect}, we present averaged information on the instance size and the reduction achieved. The average instance size can be reduced to 22\,\% using our reductions. This preprocessing takes on average \numprint{6.76} seconds. Figure~\ref{fig:reduction_effect} shows per-instance results. Here, we can see that our routine reduces 34 out of 50 instances by a factor of 2. Our preprocessing can reduce more than 50\,\% of the instances in less than a second, and overall, only eight instances take more than 1 second to reduce and are larger than 50\,\% of the original instance size.
The longest reduction time, with around 90 seconds, is spent on instance \texttt{sat14\_atco\_enc1\_opt2\_05\_4.d} (see Table~\ref{tab:reductions_overview}). This instance is one of the 10 biggest hypergraphs in our dataset. While it is not the largest among these with respect to $|H|$, it has the highest average edge size, exceeding 100. The average edge sizes of the other largest instances are mostly between 2 and 3, and not larger than 10.

\begin{obs}{Reduction Impact on Instance Size}{obs:instance_size}

    Overall, our reduction routine is highly effective at reducing instance sizes. In the majority of cases, the reduction preprocessing takes less than a second, or the size is reduced by more than 50\,\%. We can reduce the instances on average to 22\,\% of the original size, while spending 6.76 seconds.
\end{obs}

\subsection{Reduction Impact on Solver Performance}
\label{sec:experiments_solver_performance}

In this section, we analyse the effect of our reductions on solver performance.
For that, we analyze the speedups achieved by comparing a method with and without our preprocessing. A speedup of larger than 1 means our additional preprocessing resulted in a shorter overall running time compared to not using our reduction routine with that solver. 
In particular, we present experimental data for the portfolio approach (\weGotYouCovered) by Hespe~\etal~\cite{hespe2020wegotyoucovered}, the branch-and-reduce solver combined with SAT solving (\satAndReduce) by Plachetta and van der Grinten~\cite{plachetta2021sat}, and a state-of-the-art branch-and-reduce method \struction{} by Gellner~\etal~\cite{IncreasingTransformation}.
All these methods are exact solvers working on graphs. To utilize these approaches, we use the clique-expansion method, in which each hyperedge is transformed into a clique in the graph.
This clique expansion can also be performed after our reduction applications. Therefore, all solvers can benefit from our preprocessing.
For further comparisons, we also include solving the ILP using Gurobi~\cite{gurobi} on graphs (\ilpg) and hypergraphs (\ilp).
We evaluate all of these approaches on their own and in combination with our reduction preprocessing, marked with a prefix \strong{}.

Table~\ref{tab:overview} presents detailed results for all solvers on the subset of 32 solvable instances. Here, we see that in the majority of cases, using our preprocessing results leads to faster overall running time. In the following, we analyse the impact of the reduction on solver performance from different perspectives. This involves (1) comparing the benefit different solvers have from our reduction preprocessing, (2) investigating the impact on specific solvers, and (3) an instance-focused evaluation. Afterwards, we compare the best configurations for each method to identify the overall best-performing approach for solving the MIS \hbox{problem in hypergraphs.}

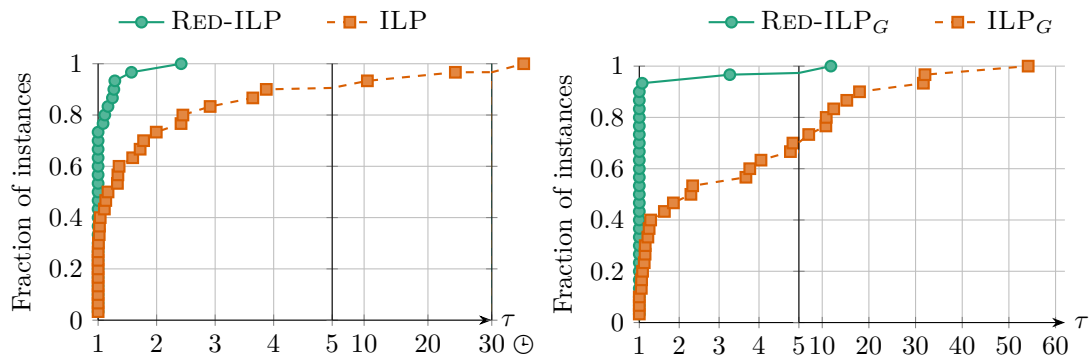
\begin{figure}[t]
\newcommand{\myPath}{data/figure_data}
\newcommand{\List}{rred_data,pred_data}
\newcommand{\ilpList}{pilp_solved,ilp_solved}
\newcommand{\ilpgList}{pilp,ilp}

\centering
\hspace{1.4cm}\ref{ilplegend}\hfill\ref{ilpglegend}
    \begin{tikzpicture}
        \begin{axis}[perf_min, perf_left, height=3.4cm, width=0.22\textwidth, xmax=5, name=p1,
                legend to name=ilplegend, myLegend3, anchor=south west, xshift=2.5cm,grid]
            \foreach \column in \ilpList {
                \addplot+[] table[x={SOLVED/\column},y=fraction, col sep=comma] {data/figure_data/ilp_time.csv};
            }
            \legend{\pilp,\ilp}
        \end{axis}
        \begin{axis}[perf_min, perf_right,height=3.4cm, width=0.15\textwidth, at={(p1.south east)},name=p3, xmin=5,grid,xtick={10,20,30},xmax=30]
            \foreach \column in \ilpList {
                \addplot+[] table[x={SOLVED/\column},y=fraction, col sep=comma] {data/figure_data/ilp_time.csv};
            }
        \end{axis}
        \begin{axis}[perf_min, perf_right, height=3.4cm, width=0.03\textwidth, at={(p3.south east)}, xtick={500}, xticklabels={\large{\clock}}, xmin=25, xmax=500,x axis line style={-},x axis line style={opacity=0}, xtick style={draw=none}, xticklabel style={anchor=north}, xlabel={}]
            \foreach \column in \ilpList {
                \addplot+[] table[x={SOLVED/\column},y=fraction, col sep=comma] {data/figure_data/ilp_time.csv};
            }
        \end{axis}
    \end{tikzpicture}\begin{tikzpicture}
        \begin{axis}[perf_min, perf_left, height=3.4cm, width=0.15\textwidth, xmax=5, name=p1,
                legend to name=ilpglegend, myLegend3, anchor=south west, xshift=2.5cm,grid]
            \foreach \column in \ilpgList {
                \addplot+[] table[x={SOLVED/\column},y=fraction, col sep=comma] {data/figure_data/ilpg_time.csv};
            }
            \legend{\pilpg,\ilpg}
        \end{axis}
        \begin{axis}[perf_min, perf_right,height=3.4cm, width=0.25\textwidth, at={(p1.south east)},name=p3, xmin=5,grid,xtick={10,20,30,40,50,60},xmax=62]
            \foreach \column in \ilpgList {
                \addplot+[] table[x={SOLVED/\column},y=fraction, col sep=comma] {data/figure_data/ilpg_time.csv};
            }
        \end{axis}
    \end{tikzpicture}
\captionof{figure}{Performance profiles comparing running times of \ilp{} (left) and \ilpg{} (right) with and without data reduction preprocessing on solvable instances.  The \clock{} symbol \hbox{indicates a timeout.}}\label{fig:ilp_reduction_effect}
\end{figure}

\subsubsection{Reduction Impact - Solver Comparison}
First, we evaluate how the different solvers react to our reduction preprocessing. All open-source solvers tested (\weGotYouCovered, \satAndReduce, \struction) work on the graph constructed via clique expansion and use data reduction rules for the maximum (weight) independent set problem in graphs. 
We present geometric mean speedups over the instances solved by all algorithms in Figure~\ref{fig:speedup_solver_comparison}.
Overall, we see that all methods benefit from our reduction routine, even though some already use reduction techniques.
When comparing the different methods, we can see varying levels of improvement. The best speedups we get are for \pilp{} at 3.84, while the portfolio approach \weGotYouCovered{} benefits the least from preprocessing.
Looking at Table~\ref{tab:reductions_overview}, we see that the largest speedup of 53 is achieved by \pilpg{} on the instance \texttt{ssmc\_dictionary28}. For this instance, it is noticeable that especially the simple \nameref{degree-one} rule already reduces 64\,\% of the vertices. 

\subsubsection{Reduction Impact - Focus on \ilp{} and \ilpg}
We now analyse the impact of the reduction on single solvers in detail. 
Table~\ref{tab:reductions_overview} shows how many instances are solved by which solver. For this section, we focus on \ilp{} and \ilpg{}, since these have 31 solved instances, almost twice as many as the open-source methods can solve.
In Figure~\ref{fig:ilp_reduction_effect}, we present performance profiles comparing different \ilp{} configurations on the left and \ilpg{} on the right. For \ilp{}, we see that for around 40\,\% of the instances, both approaches perform best (i.e., are fastest or equal). After that fraction of instances, \ilp{} performs worse than the reduction configuration. On more than 20\,\% of the instances, \ilp{} takes at least twice as long as the reduction variant, while the biggest difference is for one instance, where \ilp{} takes more than 20 times longer. This instance is \texttt{sat14\_6s130-opt}, see Table~\ref{tab:overview}. 
Furthermore, with \pilp{} we can solve the instance \texttt{sat14\_ACG-20-10p1.p}, which can not be solved by \ilp{} (see Table~\ref{tab:overview}).
For the reduction configuration, there is only one instance that is solved within a factor of 2 slower than \ilp{}. This is instance \texttt{sat14\_atco\_enc1\_opt2\_10\_12.d}, which is the instance where the reductions took around 90 seconds (see Section~\ref{sec:experiments_instance_size}). Solving this instance with \ilp{} takes around 50 seconds, so using the reduction preprocessing adds too much overhead to be beneficial.
Overall, we observe a significant improvement in the performance of \ilp{} when using our data reduction rules.
On the right in Figure~\ref{fig:ilp_reduction_effect}, we compare \ilpg{} and \pilpg{}. For this approach, we see an even clearer improvement when using our preprocessing. For around 90\,\% of the instances, the reduction variant is fastest, achieving speedups of up to 53x.

\npthousandsep{\,}
\begin{landscape}
    \begin{table}[ht]
        \centering
        \caption{Comparison of different algorithms on the set of solvable instances sorted by original hypergraph size $|H|$. The optimal solutions are presented in column $S$. If a method is unable to solve an instance, we mark it with a -; otherwise, we state the time it took to solve the instance. For each solver, we present the time $t$ (without reductions), as well as $t_{\strong}$ using our reduction preprocessing. All times are given in seconds. For solvers working on graphs, the time for the clique expansion is included. For each solver, we highlight the faster time in \textbf{bold}.}
        \setlength{\tabcolsep}{4.25pt}
\renewcommand{\arraystretch}{.86}
\begin{filecontents*}{data/table_data/exact_time.csv}
opt,graph,time_ILP,time_ILPKStrong,time_ilp_graph,time_pilp_graph,time_weGotyouGraph,time_pweGotyouGraph,time_snrGraph,time_psnrGraph,time_structionGraph,time_pstructionGraph,best,empty
,ssmc\_poli3,0.0952128,0.0702634,1.268455,0.07501335,0.1262601,0.07501335,0.917487,0.07501335,0.1303449,0.07202192,\numprint{3523},
,ssmc\_bips07\_1998,0.177471,0.069043,0.3520321,0.07432516,0.0396892,0.032942668,0.3074981,0.03596308,0.03528131,0.032920733,\numprint{5359},
,sat14\_AProVE07-27.d,6.76152,7.88231,28.5354,23.57259,19.9853,11.09051,247.1577,212.16689,20.3587,18.77829,\numprint{3165},
,ssmc\_g7jac040sc,0.753087,0.569273,13.8270088,1.3014817,0.1501547,0.1508766,2.9097588,1.0671187,0.8592938,0.8441557,\numprint{2960},
,ssmc\_psse2,0.177277,0.0484101,0.2126848,0.05379342,0.02384977,0.02685362,0.1199428,0.027199226,0.02013973,0.026953756,\numprint{3098},
,ssmc\_dictionary28,6.2093,6.37398,345.767585,6.4929685,101.788585,101.7285185,-1,-1,-1,-1,\numprint{29631},
,ssmc\_ca-CondMat,0.444052,0.323967,8.70311,0.33004344,0.4989561,0.33004344,1.74044,0.33004344,0.552708,0.32482494,\numprint{4057},
,sat14\_6s16.p,4.5773,4.97646,4.205587,4.437033,-1,-1,-1,-1,-1,-1,\numprint{11807},
,sat14\_6s11-opt,6.50335,2.34817,8.076674,2.2264811,0.820691,0.2421019,4.196094,1.0725371,0.660517,0.1713591,\numprint{30343},
,sat14\_6s184,17.1572,8.82375,16.862163,7.2479092,0.795526,0.5286702,11.021163,11.5622792,13.412763,13.9556792,\numprint{30630},
,sat14\_6s12,15.7667,11.582,17.755677,13.7598247,1.10217,0.5025687,12.015677,7.8288047,1.105075,0.6378497,\numprint{31353},
,ssmc\_lp\_pds\_20,1.08812,0.892363,4.352767,2.728206,-1,-1,-1,-1,-1,-1,\numprint{14743},
,ssmc\_ex19,0.227956,0.108866,1.620303,0.11327279,0.168173,0.11327279,0.958227,0.11327279,0.1799201,0.07501335,\numprint{985},
,ssmc\_lhr14,0.368953,0.220332,3.440477,0.3340256,0.2223171,0.201006567,1.001684,0.2323472,0.2247349,0.20161358,\numprint{4314},
,sat14\_6s133.p,6.02835,6.76662,6.017714,5.5364929,3.953914,3.7619629,-1,-1,16.181644,6.4052829,\numprint{18577},
,ispd98\_ibm11,-1,-1,2566.162028,2231.310113,-1,-1,-1,-1,-1,-1,\numprint{19660},
,sat14\_6s130-opt,73.9401,3.05092,109.236691,3.3232515,0.721928,0.3555795,4.832961,2.9047515,0.5453981,0.2678357,\numprint{45782},
,sat14\_6s130-opt.p,37.0645,44.5071,52.366948,41.0764501,1.518028,1.1979201,-1,-1,1.792558,1.1643201,\numprint{18727},
,sat14\_bob12m09-opt.d,20.0069,22.1469,86.19312,23.40556,-1,-1,-1,-1,-1,-1,\numprint{24209},
,ssmc\_garon2,1.3407,0.757206,5.459738,1.1849817,-1,-1,-1,-1,-1,-1,\numprint{400},
,sat14\_atco\_enc1\_opt2\_10\_12.d,105.146,102.209,-1,-1,-1,-1,-1,-1,-1,-1,\numprint{4594},
,ssmc\_msc10848,1.30035,1.78708,12.58058,1.86901414,1.361788,1.86948964,4.29158,1.83338097,1.667177,1.83077953,\numprint{88},
,sat14\_dated-10-11-u,154.922,146.829,41.157732,34.295623,-1,-1,-1,-1,-1,-1,\numprint{90537},
,sat14\_atco\_enc1\_opt2\_05\_4.d,50.1075,125.683,-1,-1,-1,-1,-1,-1,-1,-1,\numprint{7086},
,sat14\_manol-pipe-c10nid\_i.p,700.845,1101.02,711.535317,623.917637,-1,-1,-1,-1,-1,-1,\numprint{91911},
,sat14\_dated-10-17-u,-1,-1,1369.00664,1251.713915,-1,-1,-1,-1,-1,-1,\numprint{146629},
,sat14\_ACG-20-10p1.p,-1,3144.24,2173.7001,2074.088913,-1,-1,-1,-1,-1,-1,\numprint{113471},
,sat14\_openstacks-p30\_3.085,62.2568,52.3685,47.62701,21.27872,-1,-1,-1,-1,-1,-1,\numprint{166047},
,sat14\_q\_query\_3\_L100\_coli.p,339.262,412.356,223.97324,760.878362,-1,-1,-1,-1,-1,-1,\numprint{131181},
,ssmc\_sls,75.7554,50.893,325.57906,181.831659,1.525079,20.392671,2975.05806,1771.732659,15.79166,34.652359,\numprint{54721},
,sat14\_q\_query\_3\_L150\_coli,318.073,36.2554,412.9406,32.450234,-1,-1,-1,-1,-1,-1,\numprint{484565},
,sat14\_E02F22.p,221.358,52.387,1.727705,20.864948,-1,-1,510.531024,529.439972,-1,-1,\numprint{594},
\end{filecontents*}
\DTLloaddb{data}{data/table_data/exact_time.csv}
\hspace{-9.95cm}\begin{adjustbox}{max width=.89\textwidth, max totalheight=\textheight, keepaspectratio}
	\begin{tabular*}{\textwidth}{lrrrrrrrrrrrrrrrr}
		&
		&&\multicolumn{2}{c}{\ilp}
		&&\multicolumn{2}{c}{\ilpg}
		&&\multicolumn{2}{c}{\weGotYouCovered}
		&&\multicolumn{2}{c}{\satAndReduce}
		&&\multicolumn{2}{c}{\struction}\\
		Graph & $\alpha(H)$
		&& $t$ & $t_{\strong}$
		&& $t$ & $t_{\strong}$
		&& $t$ & $t_{\strong}$
		&& $t$ & $t_{\strong}$
		&& $t$ & $t_{\strong}$ \\
		\cmidrule{1-2}
		\cmidrule{4-5}
		\cmidrule{7-8}
		\cmidrule{10-11}
		\cmidrule{13-14}
		\cmidrule{16-17}


		\DTLforeach{data}
		{\opt=opt,\graph=graph,\tilp=time_ILP,\tpilp=time_ILPKStrong,\tilpg=time_ilp_graph,\tpilpg=time_pilp_graph,\twe=time_weGotyouGraph,\tpwe=time_pweGotyouGraph,\tsat=time_snrGraph,\tpsat=time_psnrGraph,\tstr=time_structionGraph,\tpstr=time_pstructionGraph,\S=best,\empty=empty}%
		{%
			\texttt{\graph} & \S
			&& \CompareTimesStyle{\tilp}{\tpilp} & \CompareTimesStyle{\tpilp}{\tilp}
			&& \CompareTimesStyle{\tilpg}{\tpilpg} & \CompareTimesStyle{\tpilpg}{\tilpg}
			&& \CompareTimesStyle{\twe}{\tpwe} & \CompareTimesStyle{\tpwe}{\twe}
			&& \CompareTimesStyle{\tsat}{\tpsat} & \CompareTimesStyle{\tpsat}{\tsat}
			&& \CompareTimesStyle{\tstr}{\tpstr} & \CompareTimesStyle{\tpstr}{\tstr} \\
		}
	\end{tabular*}
\end{adjustbox}
        \label{tab:overview}
    \end{table}
\end{landscape}

\subsubsection{Reduction Impact - Instance Comparison}
\begin{figure}[t]
\centering
\ref{legend_speedup2}\vspace{-.6cm}
\begin{adjustbox}{width=\textwidth}
\begin{tikzpicture}
\begin{axis}[
    ybar,
    bar width=10pt,
    ymajorgrids,
    myLegend6,
    legend to name=legend_speedup2,
    width=\textwidth,
    xtick=data,
    symbolic x coords={sat14-atco-enc1-opt2-05, sat14-q-query-3-L100, ssmc-sls, sat14-E02F22, sat14-manol-pipe, sat14-6s16, sat14-atco-enc1-opt2-10, sat14-ACG, sat14-dated-10-17, sat14-dated-10-11, ispd98-ibm11, sat14-AProVE, sat14-6s130.p, sat14-6s133, ssmc-lp-pds-20, sat14-6s184, ssmc-msc10848, sat14-6s12, sat14-openstacks, sat14-bob12m09, ssmc-g7jac040sc, ssmc-psse2, ssmc-lhr14, ssmc-bips07-1998, ssmc-garon2, sat14-6s11, ssmc-dictionary28, ssmc-ca-CondMat, ssmc-ex19, ssmc-poli3, sat14-6s130, sat14-q-query-3-L150},
    height=3cm,
    ymode=log,
    enlarge x limits=0.02,
    ytick={0.5,1,2,4,8,16},
    yticklabels={0.5,1,2,4,8,16},   
    log basis y=2,
    xticklabel style={rotate=45, anchor=east,xshift=4pt,yshift=-2pt},
	x axis line style={-},
	y axis line style={-Stealth},
    ylabel={Speedup}, ylabel style={at={(axis description cs:0.,1.05)}, anchor=south, rotate=-90},
    xlabel={}
]


\addplot[Dark2-B, fill=Dark2-B!70, postaction={ pattern=dots }, xshift=-1pt] table[x=graph, y=STRONG, col sep=tab] {data/figure_data/speedup.tsv};

\end{axis}
\end{tikzpicture}
\end{adjustbox}
\caption{Instance-wise speedup comparison of our reduction configurations.  For each instance, we compute the speedup for all solvers achieved when using our two reduction configurations compared to not using our reductions. We display the geometric mean over all methods that solved the corresponding instance. Note that we exchanged $\_$ with $-$ and abbreviated some instance names.}\label{fig:speedup_instance_comparison}
\end{figure}
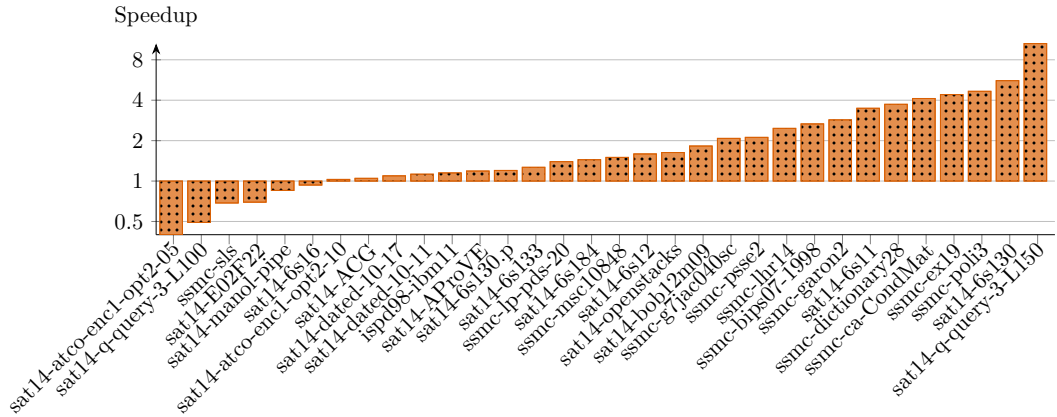
In this section, we perform an instance-focused analysis. Again, we compare the speedup achieved by our reduction techniques to the same method without preprocessing. In Figure~\ref{fig:speedup_instance_comparison}, we show for each instance the geometric mean speedup over all methods that solved that instance. Here, we see that in 6 out of 32 instances, the reduction preprocessing overall led to longer running times (speedup $< 1$). This is most noticeable for the instance \texttt{sat14\_atco\_enc1\_opt2\_10\_12.d}, where the time spent reducing the instance is longer than solving the instance right away, as already discussed above. However, when looking at Table~\ref{tab:overview} and Table~\ref{tab:reductions_overview}, we see that for the second instance \texttt{sat14\_q\_query\_3\_L100}, the reduction time is only 4.34 seconds, while solving the instance takes around \numprint{224} seconds for \ilp{} and \numprint{761} seconds for \pilpg{}. This shows that the increased running time is not due to high overhead in reduction time, but that the reduced instance is, in the end, more difficult to solve.
For the majority of instances, however, the reduction preprocessing does result in a speedup. Averaging over all instances and all solvers, we achieve a speedup of 1.9x. The best benefit for the solvers is observed for instance \texttt{sat14\_q\_query\_3\_L150}, where an 8x average speedup over all solvers is reached. 

\begin{obs}{Reduction Impact on Solver Performance}{}
 Even though some outliers are difficult to reduce or become more difficult to solve after reduction, using our preprocessing improves the overall performance of all solvers tested. The average speedup observed is 1.9x, while the greatest improvement in running time observed is for \ilpg{} with a factor of up to 53. 
\end{obs}

\subsection{Solver Comparison}

\newcommand{\List}{pilp_solved,psat_reduce,pstruction,pilp,pweGotYou}
\begin{figure}[t]
    \captionsetup[subfigure]{justification=centering}
    \centering
    \ref{legend_times_solved}
    \begin{adjustbox}{width=\textwidth}
        \begin{tikzpicture}
            \centering
            \begin{axis}[perf_min, perf_left, height=3.5cm, width=0.5\textwidth, xmax=10, name=p1, legend to name=legend_times_solved, myLegend3,restrict x to domain=1:100]
                \foreach \column in \List {
                    \addplot+[] table[x={SOLVED/\column},y=fraction, col sep=comma] {data/figure_data/exact_solved_time.csv};
                }
                \legend{\pilp,\psatAndReduce,\pstruction,\pilpg,\pweGotYouCovered}
            \end{axis}
            \begin{axis}[perf_min, perf_right,height=3.5cm, width=0.3\textwidth, xmin=10.1, at={(p1.south east)},name=p2,xmax=1000,xmode=log]
                \foreach \column in \List {
                    \addplot+[] table[x={SOLVED/\column},y=fraction, col sep=comma] {data/figure_data/exact_solved_time.csv};
                }
            \end{axis}
            \begin{axis}[perf_min, perf_mid, height=3.5cm, width=0.05\textwidth, at={(p2.south east)}, xtick={5000,5500}, xticklabels={\large{\clock}}, xmin=1000, xmax=5500,x axis line style={opacity=0}, xtick style={draw=none}, xticklabel style={anchor=north},
                ]
                \foreach \column in \List {
                    \addplot+[] table[x={SOLVED/\column},y=fraction, col sep=comma] {data/figure_data/exact_solved_time.csv};
                }
            \end{axis}
        \end{tikzpicture}
    \end{adjustbox}
    \caption{Performance profile comparing all methods combined with \strong{} on solvable instances. The method \pilp{} is the only one that works directly on hypergraphs. For all others, the clique graph transformation is performed and included in the presented time. The \clock{} symbol \hbox{indicates a timeout.}}\label{fig:soa_time_comparison}
\end{figure}
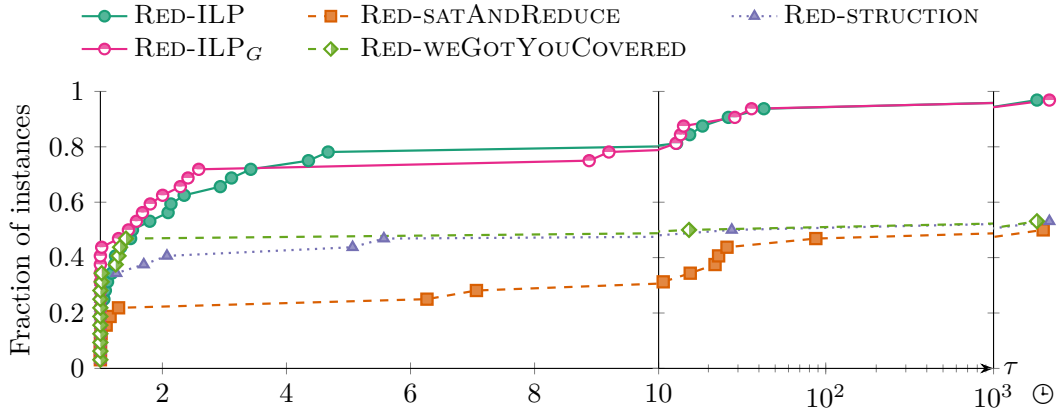

In this section, we compare the different solvers using \strong{} as a preprocessing since it is overall beneficial on the majority of instances and all solvers tested.
Figure~\ref{fig:soa_time_comparison} shows a performance profile comparing all methods on the set of solvable instances. We can see that the two Gurobi-based methods perform best, with faster execution times and the highest number of solved instances (31 out of 32). 
\pilpg{} performs slightly better compared to \pilp{} working directly on the hypergraph: 
For \pilpg{} 72\,\% of the instances can be solved within a factor of 2.5 of the running time used by the best performing algorithm on these instances, while \pilp{} reaches a factor of 3.4 at 72\,\%. The open-source methods can only solve around 53\,\% of these instances.
Among the open-source solvers, \weGotYouCovered{} can compete with the Gurobi-based methods on approximately 50\,\% of the instances and is the third-fastest method in this comparison overall.

\section{Conclusion and Future Work}
\label{sec:conclusion}

In this paper, we introduced nine new exact data reduction rules for the strong Maximum Independent Set problem in hypergraphs. With these reduction rules, our preprocessing routine can simplify a given instance. 
We evaluate the effect of our reduction routine with extensive experiments.
Our results demonstrate a significant reduction in instance size and improvements in running time for subsequent solvers. 
We compared different open-source methods and solving an ILP formulation using Gurobi, a commercial ILP solver.
The preprocessing routine reduces instances, on average, to 22\,\% of their original size while taking only 6.76 seconds. When combining our reduction preprocessing with the best-performing exact solver, we observe an average speedup of 3.84x over not using the reduction rules. On single instances, we can achieve speedups of up to a factor of 53. 
Additionally, a new instance became solvable by a method when combined \hbox{with our preprocessing.}

For future work, we are interested in developing additional data reduction rules for this problem as well as its weighted generalization. Furthermore, we would like to engineer different algorithms and heuristics for the Maximum Independent Set problem on hypergraphs, utilizing these reduction rules. It is well known that heuristic methods can also benefit from reduction preprocessing, so this is a promising direction for future work, especially on large hypergraphs, where using clique expansion to utilize solvers on graphs becomes infeasible.

\newpage

\bibliography{quellen}

\npthousandsep{\,}
\newpage
\appendix
\begin{table}
        \caption{We present the properties of the full set of instances combined with their reduction details. We give the number of vertices $n$, number of edges $m$ and the average edge size $\bar e$.}\label{tab:reductions_overview} 
    
\begin{filecontents*}{data/table_data/reductions_rn.csv}
opt,graph,n,m,e,time-reduceKFast,time-reduceKStrong,rn-reduceKFast,rm-reduceKFast,n-reduceKFast,m-reduceKFast,re-reduceKFast,re-reduceKStrong,rn-reduceKStrong,rm-reduceKStrong,n-reduceKStrong,m-reduceKStrong,best
,dac2012\_superblue6,\numprint{1011661},\numprint{1006628},3.44717,50.2825,65.2636,\numprint{715794},\numprint{658611},\numprint{1011662},\numprint{1006629},3.52856,3.52856,\numprint{538067},\numprint{551314},\numprint{1011662},\numprint{1006629},538067
,isdp98\_ibm11,\numprint{70558},\numprint{81453},4.11393,0.319856,0.394288,\numprint{54108},\numprint{54090},\numprint{70558},\numprint{81454},4.16633,4.16753,\numprint{54108},\numprint{54090},\numprint{70558},\numprint{81454},54108
,isdp98\_ibm12,\numprint{71075},\numprint{77240},3.36521,0.576767,0.654386,\numprint{44300},\numprint{40396},\numprint{71076},\numprint{77240},3.30305,3.30394,\numprint{43110},\numprint{40184},\numprint{71076},\numprint{77240},43110
,sat14\_6s11-opt,\numprint{66551},\numprint{97312},2.33332,0.0945627,0.120308,\numprint{33620},\numprint{36071},\numprint{66552},\numprint{97312},2.24862,2.24903,\numprint{26564},\numprint{35976},\numprint{66552},\numprint{97312},26564
,sat14\_6s12,\numprint{68066},\numprint{99579},2.33332,0.075209,0.0960983,\numprint{55955},\numprint{77796},\numprint{68066},\numprint{99580},2.28906,2.28906,\numprint{55955},\numprint{77796},\numprint{68066},\numprint{99580},55955
,sat14\_6s130-opt,\numprint{98653},\numprint{144361},2.33332,0.145205,0.181252,\numprint{54294},\numprint{52785},\numprint{98654},\numprint{144361},2.24926,2.24938,\numprint{39269},\numprint{52759},\numprint{98654},\numprint{144361},39269
,sat14\_6s130-opt.p,\numprint{49326},\numprint{144361},2.33332,0.0606468,0.11085,\numprint{46336},\numprint{112226},\numprint{49327},\numprint{144361},2.32809,2.81648,\numprint{46059},\numprint{45057},\numprint{49327},\numprint{144361},46059
,sat14\_6s133.p,\numprint{48215},\numprint{140968},2.33332,0.0580481,0.107373,\numprint{45641},\numprint{111738},\numprint{48215},\numprint{140968},2.33513,2.83561,\numprint{45477},\numprint{44364},\numprint{48215},\numprint{140968},45477
,sat14\_6s16.p,\numprint{31482},\numprint{91887},2.33332,0.0384737,0.0799837,\numprint{29679},\numprint{66289},\numprint{31483},\numprint{91888},2.36155,2.8306,\numprint{29623},\numprint{28855},\numprint{31483},\numprint{91888},29623
,sat14\_6s184,\numprint{66729},\numprint{97515},2.33332,0.0794089,0.105494,\numprint{47820},\numprint{64943},\numprint{66730},\numprint{97516},2.27926,2.27926,\numprint{47820},\numprint{64943},\numprint{66730},\numprint{97516},47820
,sat14\_ACG-20-10p1.d,\numprint{1632906},\numprint{381708},10.0649,6.76757,7.02257,\numprint{1037508},\numprint{368559},\numprint{1632906},\numprint{381708},5.63011,5.63011,\numprint{1037508},\numprint{368559},\numprint{1632906},\numprint{381708},1037508
,sat14\_AProVE07-27.d,\numprint{29194},\numprint{7729},9.97852,1.52903,1.52248,\numprint{17884},\numprint{7686},\numprint{29194},\numprint{7729},5.46838,5.46838,\numprint{17884},\numprint{7686},\numprint{29194},\numprint{7729},17884
,sat14\_E02F22.p,\numprint{13573},\numprint{1301187},8.80893,5.49551,19.3641,\numprint{13570},\numprint{1281092},\numprint{13574},\numprint{1301188},8.90855,10.8441,\numprint{13570},\numprint{301188},\numprint{13574},\numprint{1301188},13570
,sat14\_MD5-30-5.d,\numprint{68102},\numprint{8904},29.8827,0.766814,0.782879,\numprint{16083},\numprint{8879},\numprint{68103},\numprint{8905},5.70143,5.70143,\numprint{16083},\numprint{8879},\numprint{68103},\numprint{8905},16083
,sat14\_SAT\_dat.k75-24\_1\_r3.p,\numprint{1213331},\numprint{4754041},2.51027,3.46179,3.78974,\numprint{893532},\numprint{2202814},\numprint{1213331},\numprint{4754042},2.36252,2.44029,\numprint{850574},\numprint{1752302},\numprint{1213331},\numprint{4754042},850574
,sat14\_SAT\_dat.k85-24\_1\_r3,\numprint{2712808},\numprint{5395102},2.51031,5.0014,6.08728,\numprint{2376839},\numprint{4228418},\numprint{2712808},\numprint{5395102},2.32778,2.33266,\numprint{2261737},\numprint{4119172},\numprint{2712808},\numprint{5395102},2261737
,sat14\_SAT\_dat.k95-24\_1\_r3.p,\numprint{1540070},\numprint{6036161},2.51034,4.43476,4.78964,\numprint{1135012},\numprint{2801254},\numprint{1540071},\numprint{6036162},2.3622,2.42012,\numprint{1092054},\numprint{2350742},\numprint{1540071},\numprint{6036162},1092054
,sat14\_UCG-15-10p1,\numprint{400006},\numprint{1019220},2.38519,0.535438,0.757355,\numprint{264156},\numprint{722663},\numprint{400006},\numprint{1019221},2.36498,2.38612,\numprint{258960},\numprint{668424},\numprint{400006},\numprint{1019221},258960
,sat14\_UCG-15-10p1.p,\numprint{200003},\numprint{1019220},2.38519,0.385415,0.681908,\numprint{196530},\numprint{953836},\numprint{200003},\numprint{1019221},2.35931,2.79874,\numprint{196520},\numprint{414909},\numprint{200003},\numprint{1019221},196520
,sat14\_atco\_enc1\_opt2\_05\_4.d,\numprint{386162},\numprint{14635},112.927,56.1799,89.5522,\numprint{379637},\numprint{14391},\numprint{386163},\numprint{14636},112.729,114.819,\numprint{358895},\numprint{13351},\numprint{386163},\numprint{14636},358895
,sat14\_atco\_enc1\_opt2\_10\_12.d,\numprint{147852},\numprint{9494},65.1509,9.74993,9.7703,\numprint{140482},\numprint{9037},\numprint{147853},\numprint{9495},65.2579,65.2579,\numprint{140482},\numprint{9037},\numprint{147853},\numprint{9495},140482
,sat14\_bob12m09-opt.d,\numprint{152446},\numprint{51143},6.95499,2.91594,2.94878,\numprint{103952},\numprint{51143},\numprint{152446},\numprint{51144},4.11059,4.11059,\numprint{103952},\numprint{51143},\numprint{152446},\numprint{51144},103952
,sat14\_dated-10-11-u,\numprint{283720},\numprint{629460},2.27158,0.467133,0.805884,\numprint{128872},\numprint{149104},\numprint{283720},\numprint{629461},2.58269,2.58269,\numprint{80977},\numprint{149104},\numprint{283720},\numprint{629461},80977
,sat14\_dated-10-17-u,\numprint{459088},\numprint{1070757},2.30783,0.839878,1.23592,\numprint{212513},\numprint{305594},\numprint{459088},\numprint{1070757},2.64037,2.64037,\numprint{136195},\numprint{305594},\numprint{459088},\numprint{1070757},136195
,sat14\_manol-pipe-c10nid\_i.p,\numprint{252516},\numprint{750876},2.33333,0.350012,0.528328,\numprint{217170},\numprint{601920},\numprint{252516},\numprint{750877},2.2709,2.76144,\numprint{215809},\numprint{214147},\numprint{252516},\numprint{750877},215809
,sat14\_manol-pipe-c10nidw.d,\numprint{1291714},\numprint{433601},6.95109,47.8718,48.5203,\numprint{994621},\numprint{433296},\numprint{1291714},\numprint{433601},4.90007,4.90007,\numprint{994621},\numprint{433296},\numprint{1291714},\numprint{433601},994621
,sat14\_openstacks-p30\_3.085,\numprint{643132},\numprint{1643601},2.37843,3.41668,3.67201,\numprint{257777},\numprint{93065},\numprint{643132},\numprint{1643601},5.35343,5.35343,\numprint{133579},\numprint{93065},\numprint{643132},\numprint{1643601},133579
,sat14\_q\_query\_3\_L100\_coli.p,\numprint{315616},\numprint{1522582},2.8945,2.85629,4.34252,\numprint{250330},\numprint{1274918},\numprint{315617},\numprint{1522583},2.75469,3.08568,\numprint{223699},\numprint{592538},\numprint{315617},\numprint{1522583},223699
,sat14\_q\_query\_3\_L150\_coli,\numprint{973984},\numprint{2456707},2.90922,3.6472,4.00404,\numprint{283996},\numprint{379759},\numprint{973984},\numprint{2456708},2.21531,2.21565,\numprint{108840},\numprint{377131},\numprint{973984},\numprint{2456708},108840
,ssmc\_ABACUS\_shell\_hd,\numprint{23411},\numprint{23411},9.33214,0.0929969,0.107097,\numprint{8976},\numprint{9006},\numprint{23412},\numprint{23412},6.66189,6.66189,\numprint{8976},\numprint{9006},\numprint{23412},\numprint{23412},8976
,ssmc\_BenElechi1,\numprint{245873},\numprint{245873},53.4847,6.75232,6.92756,\numprint{25703},\numprint{164220},\numprint{245874},\numprint{245874},7.13712,7.90685,\numprint{24533},\numprint{24605},\numprint{245874},\numprint{245874},24533
,ssmc\_NotreDame\_actors,\numprint{127823},\numprint{383640},3.83277,10.0665,10.051,\numprint{17688},\numprint{16771},\numprint{127823},\numprint{383640},3.39175,3.43879,\numprint{11656},\numprint{16190},\numprint{127823},\numprint{383640},11656
,ssmc\_bips07\_1998,\numprint{15065},\numprint{15065},4.12837,0.0252834,0.0282382,\numprint{459},\numprint{448},\numprint{15066},\numprint{15066},2.83482,2.83482,\numprint{459},\numprint{448},\numprint{15066},\numprint{15066},459
,ssmc\_ca-CondMat,\numprint{23132},\numprint{23132},8.08092,0.319161,0.324273,\numprint{0},\numprint{0},\numprint{23133},\numprint{23133},0,0,\numprint{0},\numprint{0},\numprint{23133},\numprint{23133},0
,ssmc\_dictionary28,\numprint{52651},\numprint{39326},4.52809,0.168802,0.186003,\numprint{4548},\numprint{6100},\numprint{52652},\numprint{39327},3.98803,3.98803,\numprint{4548},\numprint{6100},\numprint{52652},\numprint{39327},4548
,ssmc\_ex19,\numprint{12004},\numprint{12004},21.6476,0.105671,0.109396,\numprint{0},\numprint{0},\numprint{12005},\numprint{12005},0,0,\numprint{0},\numprint{0},\numprint{12005},\numprint{12005},0
,ssmc\_fd18,\numprint{16428},\numprint{16428},3.85963,0.0172979,0.0326585,\numprint{13712},\numprint{13652},\numprint{16428},\numprint{16428},3.92836,3.92836,\numprint{13712},\numprint{13652},\numprint{16428},\numprint{16428},13712
,ssmc\_g7jac040sc,\numprint{11789},\numprint{11789},9.72612,0.0794039,0.0896445,\numprint{6501},\numprint{6923},\numprint{11790},\numprint{11790},4.70605,4.70605,\numprint{6501},\numprint{6923},\numprint{11790},\numprint{11790},6501
,ssmc\_garon2,\numprint{13535},\numprint{13535},28.859,0.188716,0.196093,\numprint{3666},\numprint{1456},\numprint{13535},\numprint{13535},9.85852,9.85852,\numprint{3666},\numprint{1456},\numprint{13535},\numprint{13535},3666
,ssmc\_lhr14,\numprint{14269},\numprint{14269},21.5738,0.182322,0.188948,\numprint{692},\numprint{544},\numprint{14270},\numprint{14270},23.4798,23.4798,\numprint{692},\numprint{544},\numprint{14270},\numprint{14270},692
,ssmc\_lp\_pds\_20,\numprint{108174},\numprint{33798},6.88345,0.290964,0.324035,\numprint{66574},\numprint{17781},\numprint{108175},\numprint{33798},7.48822,7.48822,\numprint{66574},\numprint{17781},\numprint{108175},\numprint{33798},66574
,ssmc\_m14b,\numprint{214765},\numprint{214765},15.6359,2.91748,2.9144,\numprint{182021},\numprint{199595},\numprint{214765},\numprint{214765},13.3087,13.3087,\numprint{182021},\numprint{199595},\numprint{214765},\numprint{214765},182021
,ssmc\_msc10848,\numprint{10848},\numprint{10848},113.364,1.74347,1.80166,\numprint{182},\numprint{158},\numprint{10848},\numprint{10848},3.53797,3.53797,\numprint{182},\numprint{158},\numprint{10848},\numprint{10848},182
,ssmc\_pdb1HYS,\numprint{36417},\numprint{36417},119.306,5.15839,5.3907,\numprint{10907},\numprint{16139},\numprint{36417},\numprint{36417},39.4991,41.0618,\numprint{10859},\numprint{9897},\numprint{36417},\numprint{36417},10859
,ssmc\_poli3,\numprint{16955},\numprint{16955},2.23232,0.0684567,0.0701345,\numprint{0},\numprint{0},\numprint{16955},\numprint{16955},0,0,\numprint{0},\numprint{0},\numprint{16955},\numprint{16955},0
,ssmc\_psse2,\numprint{11028},\numprint{28633},4.02535,0.0207842,0.0220501,\numprint{63},\numprint{57},\numprint{11028},\numprint{28634},2.21053,2.21053,\numprint{63},\numprint{57},\numprint{11028},\numprint{28634},63
,ssmc\_rgg\_n\_2\_18\_s0,\numprint{262143},\numprint{262141},11.805,2.19212,2.56876,\numprint{161198},\numprint{238657},\numprint{262144},\numprint{262141},7.09219,7.50918,\numprint{161031},\numprint{178535},\numprint{262144},\numprint{262141},161031
,ssmc\_ship\_001,\numprint{34920},\numprint{34920},132.996,7.47016,7.76614,\numprint{4347},\numprint{14346},\numprint{34920},\numprint{34920},18.5822,26.1299,\numprint{4347},\numprint{1547},\numprint{34920},\numprint{34920},4347
,ssmc\_sls,\numprint{62729},\numprint{1748121},3.89235,13.7847,19.5895,\numprint{61865},\numprint{1561682},\numprint{62729},\numprint{1748122},2.00071,2.00054,\numprint{60640},\numprint{1362532},\numprint{62729},\numprint{1748122},60640
,ssmc\_xenon2,\numprint{157464},\numprint{157464},24.556,1.70785,1.84799,\numprint{51691},\numprint{132533},\numprint{157464},\numprint{157464},7.97687,7.89351,\numprint{51670},\numprint{58260},\numprint{157464},\numprint{157464},51670
,mean,\numprint{315399},\numprint{676599},18.1824856,5.235120324,6.756562956,\numprint{213376},\numprint{391249},\numprint{315400},\numprint{676600},8.741439,9.0848336,\numprint{195602},\numprint{302658},\numprint{315400},\numprint{676600},
\end{filecontents*}
\DTLloaddb[
       keys={opt,graph,n,m,e,tKf,tKs,rnKf,rmKf,nKf,mKf,reKf,reKs,rnKs,rmKs,nKs,mKs,best}
]{reductions}{data/table_data/reductions_rn.csv}

\begin{adjustbox}{max width=.78\textwidth, max totalheight=\textheight, keepaspectratio}
       \begin{tabular*}{\textwidth}{lrrrrrrrr}
              & \multicolumn{3}{c}{Original Graph}
              &                 & \multicolumn{4}{c}{\strong}                                                                                                                                                                                                                      \\
              \cmidrule{1-4} \cmidrule{6-9}
              & $n$             & $m$                         & $\bar e$
              &                 & $n_r$                       & $m_r$                                 & $\bar e_r$                                     & $t$                                                                                                                       \\
              \cmidrule{1-4} \cmidrule{6-9}

              \DTLforeach{reductions}
              {\graph=graph,\n=n,\m=m,\e=e, \rnKs=rnKs,\rmKs=rmKs,\reKs=reKs,\tKs=tKs}%
              {%
                     \DTLiflastrow{
                            \phantom{.}\\[-0.8em]
                            \cmidrule{1-4} \cmidrule{6-9}
                            \phantom{.}\\[-0.8em]
                     }{}

                     \graph & \n              & \m                          & \dtlround{\erounded}{\e}{2} \erounded
                     & \hspace*{0.2cm} & \rnKs                       & \rmKs                                 & \dtlround{\reKsrounded}{\reKs}{2} \reKsrounded & \dtlround{\tKsrounded}{\tKs}{2} \tKsrounded                                                                               \\
              }
       \end{tabular*}
\end{adjustbox}

\end{table}

\end{document}